\documentclass{article}
\usepackage{amsmath}
\usepackage{amsfonts}

\newtheorem{theorem}{Theorem}

\newtheorem{corollary}[theorem]{Corollary}

\newtheorem{definition}[theorem]{Definition}

\newtheorem{lemma}[theorem]{Lemma}

\newtheorem{proposition}[theorem]{Proposition}

\newenvironment{proof}[1][Proof]{\noindent \textbf{#1.} }{\  \rule{0.5em}{0.5em}}
\input{tcilatex}
\begin{document}

\title{Temporal and spatial turnpike-type results under forward
time-monotone performance criteria\thanks{%
This work was presented at seminars and workshops at ETH, King's College,
Princeton, Oxford and Columbia University. The authors would like to thank
the participants for fruitful comments and suggestions.}}
\author{T. Geng\thanks{%
Department of Mathematics, The University of Texas at Austin;
tgeng@math.utexas.edu.} \ and T. Zariphopoulou\thanks{%
Departments of Mathematics and IROM, The University of Texas at Austin and
the Oxford-Man Institute of Quantitative Finance, University of Oxford;
zariphop@math.utexas.edu.}}
\date{First draft: November 2016, this draft: February 2017}
\maketitle

\begin{abstract}
We present turnpike-type results for the risk tolerance function in an
incomplete market It\^{o}-diffusion setting under time-monotone forward
performance criteria. We show that, contrary to the classical case, the
temporal and spatial limits do not coincide. Rather, they depend directly on
the left- and right-end of the support of an underlying measure associated
with the forward performance criterion. We present examples with discrete
and continuous such measures, and discuss the asymptotic behavior of the
risk tolerance for each case.
\end{abstract}

\section{Introduction}

Turnpike results in maximal expected utility models yield the behavior of
optimal portfolio functions when the investment horizon is long, under
asymptotic assumptions on the investor's risk preferences.

The essence of the "turnpike" result (stated, for simplicity, for a single
log-normal stock with coefficients $\mu $ and $\sigma $) is the following:
assume that the investment horizon is $\left[ 0,T\right] $ and that the
investor's utility $U_{T}$ behaves like a power function for large wealth
levels, i.e., 
\begin{equation}
U_{T}\left( x\right) \sim \frac{1}{\gamma }x^{\gamma },\text{ \  \  \  \ }x%
\text{ \  \ large.}  \label{utility-classical}
\end{equation}%
Then, if this horizon is very long, the associated optimal portfolio
function $\pi ^{\ast }\left( x,t;T\right) $ is "close" to the one
corresponding to this power utility, i.e., for \textit{each} $x>0,$ $t\in %
\left[ 0,T\right] ,$

\begin{equation}
\frac{\pi ^{\ast }\left( x,t;T\right) }{x}\sim \frac{\mu }{\sigma ^{2}}\frac{%
1}{1-\gamma },\text{ \  \  \ }T\text{ \  \ large.}  \label{turnpike-classical}
\end{equation}

In other words, the asymptotic \textit{spatial} behavior of the terminal
datum dictates the long-term \textit{temporal} behavior of the portfolio
function for \textit{every} wealth level.

We recall that the function $\pi ^{\ast }\left( x,t;T\right) $ is the one
the determines the optimal wealth process in feedback form, in that the
optimal wealth process $X_{t}^{\ast },$ $t\in \left[ 0,T\right] ,$ is
generated by the investment strategy $\pi _{t}^{\ast }=\pi ^{\ast }\left(
X_{t}^{\ast },t;T\right) .$

Turnpike results can be found in \cite{CH92} where a continuous-time model
was first considered, and the turnpike properties were established using
contingent claim methods. Their results were later extended in \cite{HZ99}
using an autonomous equation that the function $\pi \left( x,t;T\right) $
satisfies and arguments from viscosity solutions. Duality methods were used
in \cite{DRB99} for complete markets and the incomplete market case was
studied in \cite{GKRX14}.

More recently, the authors of \cite{B-Zheng} established the rate of
convergence in a log-normal model, showing that there exist a positive
constant $c$ and a function $D\left( x\right) ,$ such that, for all $x>0,$ 
\begin{equation*}
\left \vert \pi ^{\ast }\left( x,t;T\right) -\frac{\mu }{\sigma ^{2}}\frac{1%
}{1-\gamma }x\right \vert \leq D\left( x\right) e^{-c\left( T-t\right) }.
\end{equation*}

A closer look at these turnpike results yields that we are essentially
working in a \textit{single }investment horizon setting, $\left[ 0,T\right]
, $ which is taken to be very long. As a result, however, one needs to
commit to a market model for this long horizon, but this choice cannot be
modified later on, if time-consistency is desired. Furthermore, one
pre-commits at initial time to a utility function for very far in the
future, $T.$ We also remark that no matter how big $T$ is, the optimal
investment problem is not defined beyond this point, because the utility
function is given for $T$ only.

Herein, we take an alternative point of view. Instead of committing to a
single long horizon $\left[ 0,T\right] $, we define an investment problem
for all times $t\in \left[ 0,\infty \right) $. Moreover, instead of choosing
at an initial time the utility $U_{T}$ for the remote horizon $T,$ we choose
the utility at this initial time. We also depart from the log-normal setting
and work with a general Ito-diffusion multi-security incomplete market model.

We measure the performance of investment strategies via the so-called
forward investment performance criterion. This criterion was introduced by
Musiela and one of the authors in \cite{Musiela-first} and offers
flexibility for performance measurement and risk management under model
adaptation and ambiguity, alternative market views, rolling horizons, and
others. We recall its definition and refer the reader to, among others, \cite%
{MZ-QF}, \cite{MZ3}, \textbf{\ }for an overview of the forward approach.

Herein, we focus on the class of time-monotone forward performance criteria,
studied in \cite{MZ10a} and briefly reviewed in the next section. They are
given by a time-decreasing and adapted to the market information process, $%
U\left( x,t\right) ,$ $\left( x,t\right) \in \mathbb{R}_{+}\times \left[
0,\infty \right) ,$ of the form

\begin{equation*}
U\left( x,t\right) =u\left( x,A_{t}\right) ,\text{ }
\end{equation*}%
where $u\left( x,t\right) $ is a deterministic function (see (\ref{u pde}))
and $A_{t}=\int_{0}^{t}\left \vert \lambda _{s}\right \vert ^{2}ds,$ with
the process $\lambda _{t}$ being the market price of risk. Note that $%
U\left( x,t\right) $ is a compilation of a deterministic investor-specific
input, $u\left( x,t\right) ,$ and a stochastic market-specific input, $%
A_{t}. $

The optimal investment process $\pi _{t}^{\ast }$ turns out to be, for $%
t\geq 0,$ 
\begin{equation}
\pi _{t}^{\ast }=\sigma _{t}^{+}\lambda _{t}r\left( X_{t}^{\ast
},A_{t}\right) \text{ \  \  \  \ with \  \  \ }r\left( x,t\right) :=-\frac{%
u_{x}\left( x,t\right) }{u_{xx}(x,t)},  \label{r-intro}
\end{equation}%
where $\sigma _{t}^{+}$ is the pseudo-inverse of the volatility matrix, and $%
X_{t}^{\ast },$ $t\geq 0,$ the optimal wealth generated by this investment
strategy $\pi _{t}^{\ast }$ (cf. (\ref{wealth})). The function $r\left(
x,t\right) $ is the (local) risk tolerance and will be the main object of
study herein.

Contrary to the classical case, in which a terminal datum is pre-assigned
for $T$ and the solution is then constructed for $t\in \left[ 0,T\right) ,$
in the forward setting, the criterion is defined for all times, starting
with an initial (and not terminal) datum $u_{0}\left( x\right) =U\left(
x,0\right) .$

In analogy to the classical turnpike setting, we are thus motivated to study
the following question: if the initial condition $u_{0}\left( x\right) $ is
such that 
\begin{equation}
u_{0}\left( x\right) \sim \frac{1}{\gamma }x^{\gamma },\text{ \  \  \  \ }x%
\text{ \ large,}  \label{assumption}
\end{equation}%
does this imply that, for each $x>0,$ 
\begin{equation*}
\frac{r(x,t)}{x}\sim \frac{1}{1-\gamma },\  \  \ t\text{ \ large \ }?
\end{equation*}

There are fundamental differences between the classical and the forward
settings, for one is not a mere variation of the other by a time reversal.
Rather, the classical problem is well-posed while the forward is an inverse
problem. Naturally, various properties used for the classical turnpike
results fail, with the most important being the lack of comparison principle
for various PDEs (cf. (\ref{u pde}) and (\ref{r-eqn})) at hand.

The first striking difference between the two settings is the distinct
nature of the temporal and spatial limits. Indeed, in the traditional
turnpike results in \cite{HZ99} and \cite{B-Zheng}, \ the temporal limit in (%
\ref{turnpike-classical}) coincides with the spatial one, in that for fixed
time $T_{0}$ and wealth level $x_{0},$ 
\begin{equation*}
\lim_{x\uparrow \infty }\frac{\pi \left( x,t;T_{0}\right) }{x}%
=\lim_{T\uparrow \infty }\frac{\pi \left( x_{0},t;T\right) }{x_{0}}.
\end{equation*}

However, this is \textit{not }the case in the forward setting. Indeed, the
temporal and spatial limits of the function $\frac{r\left( x,t\right) }{x}$
do \textit{not }\ coincide. This can be seen, for instance, in the
motivational example in section 2.1.

The aim herein then becomes the study of the \textit{spatial }and \textit{%
temporal} limits 
\begin{equation}
\lim_{x\uparrow \infty }\frac{r(x,t_{0})}{x}\text{ \  \  \  \ and \  \  \ }%
\lim_{t\uparrow \infty }\frac{r(x_{0},t)}{x},  \label{limits}
\end{equation}%
for fixed $t_{0}>0,x_{0}>0,$ respectively, under appropriate conditions for
the asymptotic behavior of the initial datum $u_{0}\left( x\right) ,$ for
large $x.$

Pivotal role for determining these limits is played by an underlying
positive finite Borel measure, $\mu ,$ which is the defining element for the
construction of the forward performance process. Indeed, it was shown in 
\cite{MZ10a} that the above function $u$ is uniquely (up to an additive
constant) related to a harmonic function $h:\mathbb{R\times }\left[ 0,\infty
\right) \longrightarrow \mathbb{R}^{+}$, and, furthermore, the latter is
uniquely characterized by an integral transform, specifically,%
\begin{equation}
u_{x}\left( h\left( z,t\right) ,t\right) =-e^{-x+\frac{t}{2}}\text{ \  \  \  \
with \  \  \ }h\left( z,t\right) =\int_{a}^{b}e^{zy-\frac{1}{2}y^{2}t}\mu
\left( dy\right) ,  \label{u-h}
\end{equation}%
for $0\leq a\leq b\leq \infty .$

An immediate consequence of this general solution is that the initial datum
is directly related to this measure $\mu ,$ in that $\left( u_{0}^{\prime
}\right) ^{\left( -1\right) }$ needs to be of the integral form 
\begin{equation*}
\left( u_{0}^{\prime }\right) ^{\left( -1\right) }\left( x\right)
=\int_{a}^{b}x^{-y}\mu \left( dy\right) .
\end{equation*}

As a result, it is natural to expect that the asymptotic properties of $%
u_{0}\left( x\right) ,$ which enter in the turnpike assumptions, are also
directly linked to the form and properties of $\mu $.

Furthermore, this measure also appears in the specification of the risk
tolerance function. Indeed, we deduce from (\ref{r-intro}) and (\ref{u-h})
that $r\left( x,t\right) $ can be represented as%
\begin{equation}
r\left( x,t\right) =h_{x}\left( h^{\left( -1\right) }\left( x,t\right)
,t\right) ,  \label{r-h-intro}
\end{equation}%
with both $h_{x}$ and $h^{\left( -1\right) }$ depending on $\mu .$

The main results herein are that, if the support of the measure is finite, $%
b<\infty ,$ then the \textit{spatial} limit coincides with the \textit{%
right-end} point of the support while the \textit{temporal }limit with the 
\textit{left-end} one, namely,%
\begin{equation}
\lim_{x\uparrow \infty }\frac{r(x,t_{0})}{x}=b\text{ \  \  \  \  \ and \  \  \  \  \ 
}\lim_{t\uparrow \infty }\frac{r(x_{0},t)}{x}=a.  \label{results}
\end{equation}

The first step in obtaining the above limits is to understand the connection
between the asymptotic behavior of the initial (marginal) datum and the
finiteness of the measure's support. We study the following two cases, which
correspond to the spatial and temporal limits, respectively.

We first show that the asymptotic assumption (\ref{assumption}), stated in
terms of the marginal, 
\begin{equation}
u_{0}^{\prime }\left( x\right) \  \text{\ }\sim \text{ }\ x^{\gamma -1},
\label{one}
\end{equation}%
if and only if the right end of the measure's support satisfies both $b=%
\frac{1}{1-\gamma }$ and $\mu \left( \left \{ b\right \} \right) =1.$ In
other words, condition (\ref{one}) implies that the measure must have finite
support with its right boundary equal to $\frac{1}{1-\gamma }$ and,
furthermore, with a mass at this point. Conversely, for the measure to have
these properties, condition (\ref{one}) must hold. We then establish the
first limit in (\ref{results}) using representation (\ref{u-h}), the
equation (\ref{u pde}) satisfied by $u\left( x,t\right) ,$ and various
convexity properties of $h$ and its derivatives. We stress that the
requirement that $\mu \left( \left \{ b\right \} \right) \neq 0$ cannot be
relaxed. Indeed, we show in Example 6.2, where the measure is the Lebesgue
one, that the spatial turnpike property fails.

For the second case, we relate the finiteness of the measure's support with
a relaxed version of (\ref{one}). We show that if there exists $\gamma <1,$\ 
$\gamma \neq 0$, such that for all $\gamma ^{\prime }\in \left( \gamma
,1\right) $ and $\gamma ^{^{\prime \prime }}<\gamma ,$ 
\begin{equation}
\lim_{x\uparrow \infty }\frac{u_{0}^{\prime }\left( x\right) }{x^{\gamma
^{\prime }-1}}=0\text{ \  \  \ and \  \  \ }\lim_{x\uparrow \infty }\frac{%
u_{0}^{\prime }\left( x\right) }{x^{\gamma ^{^{\prime \prime }}-1}}=\infty ,
\label{Variation}
\end{equation}%
then the right boundary of the measure's support must satisfy $b=\frac{1}{%
1-\gamma },$ and vice-versa. This regular variation assumption is weaker
than (\ref{one}), needed for the spatial limit and, naturally, yields a
weaker result. Indeed, while the support has to be finite with right
boundary equal to $\frac{1}{1-\gamma },$ it does not need to have a mass at $%
\frac{1}{1-\gamma }.$

We in turn establish the second limit in (\ref{results}), which is the
genuine analogue of the classical turnpike result. Obtaining this limit is
considerably more challenging than in the classical case due to the
ill-posed nature of the problem. Indeed, the methodology used in \cite{HZ99}
is inapplicable because of lack of comparison results for the ergodic
version of the equation satisfied by $r\left( x,t\right) .$ The approach of 
\cite{B-Zheng} does not apply either because of the lack of connection
between the solutions of the ill-posed heat equation and Feynman-Kac type
stochastic representation of its solution. Therefore, one needs to work
directly with the function $r\left( x,t\right) ,$ which, from (\ref%
{r-h-intro}) and (\ref{u-h}), is given in the implicit form%
\begin{equation*}
r\left( x,t\right) =\int_{a}^{b}ye^{yh^{\left( -1\right) }\left( x,t\right) -%
\frac{1}{2}y^{2}t}\mu \left( dy\right) ,
\end{equation*}%
where however the spatial inverse $h^{\left( -1\right) }$ is involved.

The key step in obtaining the temporal limit is to show that 
\begin{equation*}
\lim_{t\uparrow \infty }\frac{h^{\left( -1\right) }\left( x,t\right) }{t}=%
\frac{a}{2},
\end{equation*}%
where $a$ is the left end point of the measure's support. Then the temporal
convergence in (\ref{results}) and the rate of convergence is shown using
the implicit representation 
\begin{equation*}
r\left( x,t\right) -ax=\int_{a}^{b}\left( y-a\right) e^{ty\left( \frac{%
h^{\left( -1\right) }\left( x,t\right) }{t}-\frac{1}{2}y\right) }\mu \left(
dy\right) .
\end{equation*}

In addition to the general spatial and temporal convergence results, we
present two representative examples. In the first, the measure is a finite
sum of Dirac functions while, in the second, it is taken to be the Lebesgue
measure. We calculate the limits of (\ref{results}), and also provide
asymptotic expansions for the risk tolerance function.

The paper is structured as follows. In section 2, we present the market
model, the investment performance criterion and a motivating example
demonstrating that the temporal and spatial limits do not in general
coincide. In sections 3 and 4, we analyze respectively the spatial and
temporal asymptotic behavior of the relative risk tolerance, while in
section 5 we analyze the asymptotic properties of the relative prudence
function. In section 6 we present the two representative examples, and
conclude in section 7 with future research directions.

\section{The model and the investment performance criterion}

The market environment consists of one riskless and $k$ risky securities.
The prices of the risky securities are modelled as It\^{o}-diffusion
processes, namely, the price $S^{i}$ of the $i^{th}$ risky asset follows 
\begin{equation*}
dS_{t}^{i}=S_{t}^{i}\left( \mu _{t}^{i}dt+\Sigma _{j=1}^{d}\sigma
_{t}^{ji}dW_{t}^{j}\right) ,
\end{equation*}%
with $S_{0}^{i}>0,$ for $i=1,...,k.$ The process $W_{t}=\left(
W_{t}^{1},...,W_{t}^{d}\right) ,$ $t\geq 0,$ is a standard Brownian motion,
defined on a filtered probability space $\left( \Omega ,\mathcal{F},\mathbb{P%
}\right) .$ The coefficients $\mu _{t}^{i}$ and $\sigma _{t}^{i}=\left(
\sigma ^{1i},...,\sigma _{t}^{di}\right) ,$ $i=1,...,k,$ $t\geq 0,$ are $%
\mathcal{F}_{t}-$adapted processes and values in $\mathbb{R}$ and $\mathbb{R}%
^{d},$ respectively. We denote by $\sigma _{t}$ the volatility matrix, i.e.
the $d\times k$ random matrix $\left( \sigma _{t}^{ji}\right) ,$ whose $%
i^{th}$ column represents the volatility $\sigma _{t}^{i}$ of the $i^{th}$
asset. We may, then, alternatively, write the above equation as%
\begin{equation*}
dS_{t}^{i}=S_{t}^{i}\left( \mu _{t}^{i}dt+\sigma _{t}^{i}\cdot dW_{t}\right)
.
\end{equation*}%
The riskless asset, the savings account, has price process $B$ satisfying $%
dB_{t}=r_{t}B_{t}dt$ with $B_{0}=1,$ and for a nonnegative $\mathcal{F}_{t}-$%
adapted interest rate process $r_{t}.$ Also, we denote by $\mu _{t}$ the $k$%
-dimensional vector with coordinates $\mu _{t}^{i}$ and by $\mathbf{1}$%
\textbf{\ }the $k$-dim vector with every component equal to one. The
processes $\mu _{t},\sigma _{t}$ and $r_{t}$ satisfy the appropriate
integrability conditions.

We assume that $\mu _{t}-r_{t}\mathbf{1\in }\mathit{Lin}\left( \sigma
_{t}^{T}\right) ,$ where $\mathit{Lin}\left( \sigma _{t}^{T}\right) $
denotes the linear space generated by the columns of $\sigma _{t}^{T}.$
Therefore, the equation $\sigma _{t}^{T}z=\mu _{t}-r_{t}\mathbf{1}$ has a
solution, known as the market price of risk, 
\begin{equation}
\lambda _{t}=\left( \sigma _{t}^{T}\right) ^{+}\left( \mu _{t}-r_{t}\mathbf{1%
}\right) .  \label{market-price-risk}
\end{equation}%
It is assumed that there exists a deterministic constant $c>0,$ such that $%
\left \vert \lambda _{t}\right \vert \leq c$ and that $\lim_{t\uparrow
\infty }\int_{0}^{t}\left \vert \lambda _{s}\right \vert ^{2}ds=\infty .$

Starting at $t=0$ with an initial endowment $x\geq 0,$ the investor invests
at any time $t>0$ in the risky and riskless assets. The present value of the
amounts invested are denoted by the processes $\pi _{t}^{0}$ and $\pi
_{t}^{i},$ $i=1,...,k,$ respectively, which are taken to be self-financing.
The present value of her investment is then given by the discounted wealth
process $X_{t}^{\pi }=\sum \pi _{t}^{i},$ $t>0,$ which solves%
\begin{equation}
dX_{t}^{\pi }=\sigma _{t}\pi _{t}\cdot \left( \lambda _{t}dt+dW_{t}\right)
\label{wealth}
\end{equation}%
with the (column) vector $\pi _{t}=\left( \pi _{t}^{i};i=1,...,k\right) .$
It is taken to satisfy the non-negativity constraint $X_{t}^{\pi }\geq 0,$ $%
t>0.$

The set of admissible policies is given by 
\begin{equation*}
\mathcal{A}=\left \{ \pi :\text{self-financing, \ }\pi _{t}\in \mathcal{F}%
_{t},\text{ }E_{\mathbb{P}}\int_{0}^{t}\left \vert \sigma _{s}\pi _{s}\right
\vert ^{2}ds<\infty ,\text{ }X_{t}^{\pi }\geq 0,\text{ }t>0\right \} .
\end{equation*}%
The performance of admissible investment strategies is evaluated via the
so-called forward investment performance criteria, introduced in \cite%
{Musiela-first} (see, also \cite{Musiela-Carmona}, \cite{MZ-QF} and \cite%
{MZ3})\textbf{. }We review their definition next.

We introduce the domain notation $\mathbb{D}_{+}=\mathbb{R}_{+}\times \left[
0,\infty \right) $ and $\mathbb{D}=\mathbb{R}\times \left[ 0,\infty \right)
. $

\begin{definition}
An $\mathcal{F}_{t}$-adapted process $U(x,t)$ is a forward investment
performance if for $\left( x,t\right) \in \mathbb{D},$

i) the mapping $x\rightarrow U(x,t)$ is strictly increasing and strictly
concave;

ii) for each $\pi \in \mathcal{A}$, $E_{\mathbb{P}}(U(X_{t}^{\pi
},t))^{+}<\infty $, and for $s\geq t$, 
\begin{equation*}
U\left( X_{t}^{\pi },t\right) \geq E_{\mathbb{P}}\left( \left. U(X_{s}^{\pi
},s)\right \vert \mathcal{F}_{t}\right) ,
\end{equation*}

iii)\ there exists $\pi ^{\ast }\in \mathcal{A}$ such that for $s\geq t$, 
\begin{equation*}
U\left( X_{t}^{\pi ^{\ast }},t\right) =E_{\mathbb{P}}\left( \left.
U(X_{s}^{\pi ^{\ast }},s)\right \vert \mathcal{F}_{t}\right) .
\end{equation*}
\end{definition}

Herein we focus on the class of\textit{\ time-monotone} forward performance
processes. For the reader's convenience, we rewrite some of the results we
stated in the introduction. Time-monotone forward processes were extensively
studied in \cite{MZ10a}, and are given by 
\begin{equation}
U(x,t)=u(x,A_{t}),  \label{U def}
\end{equation}%
where $u:\mathbb{D}_{+}\rightarrow \mathbb{R}_{+}$ is strictly increasing
and strictly concave in $x,$ satisfying 
\begin{equation}
u_{t}=\frac{1}{2}\frac{u_{x}^{2}}{u_{xx}}.  \label{u pde}
\end{equation}%
The market input processes $A_{t}$ and $M_{t}$, $t\geq 0$, are defined as 
\begin{equation}
M_{t}=\int_{0}^{t}\lambda _{s}\cdot dW_{s}\text{ \  \  \ and \  \  \ }%
A_{t}=\int_{0}^{t}\left \vert \lambda _{s}\right \vert ^{2}ds=\left \langle
M\right \rangle _{t}.~  \label{market input}
\end{equation}%
The optimal portfolio process $\pi _{t}^{\ast }$ is given by $\pi _{t}^{\ast
}=\sigma _{t}^{+}\lambda _{t}r(X_{t}^{\ast },A_{t}),$ where the (local) risk
tolerance function $r\left( x,t\right) :\mathbb{D}_{+}\rightarrow \mathbb{R}%
_{+}$ is defined as%
\begin{equation}
r\left( x,t\right) :=-\frac{u_{x}\left( x,t\right) }{u_{xx}\left( x,t\right) 
}.  \label{risk-tolerance}
\end{equation}

Central role in the construction of the performance criterion, the optimal
policies and their wealth plays a harmonic function $h:\mathbb{D}\rightarrow 
\mathbb{R}_{+}$, defined via the transformation 
\begin{equation}
u_{x}(h(z,t),t)=e^{-z+\frac{t}{2}}.  \label{exp tran}
\end{equation}%
It solves, as it follows from (\ref{u pde}) and (\ref{exp tran}), the
ill-posed heat equation 
\begin{equation}
h_{t}+\frac{1}{2}h_{zz}=0.  \label{h pde}
\end{equation}%
Moreover, it is positive and strictly increasing in $z.$ It was shown in 
\cite{MZ10a}, that such solutions are \textit{uniquely} represented by 
\begin{equation*}
h(z,t)=\int_{a}^{b}\frac{e^{yz-\frac{1}{2}y^{2}t}-1}{y}\nu (dy)+C,
\end{equation*}%
where $a=0^{+}$ or $a>0,b\leq \infty $ and $C$ a generic constant.

The measure $\nu $ is defined on $\mathcal{B}^{+}(\mathbb{R}),$ the set of
positive Borel measures, with the additional properties that, for $z\in 
\mathbb{R},$ $\int_{a}^{b}e^{yz}\nu (dy)<\infty $ \ and $\int_{a}^{b}\frac{%
\nu (dy)}{y}<\infty .$ To simplify the presentation and without loss of
generality, we choose $C:=\int_{a}^{b}\frac{1}{y}\nu (dy)$ and, also,
introduce the normalized measure $\mu \left( dy\right) =\frac{1}{y}\nu (dy).$

Then, the function $h$ has, for $\left( z,t\right) \in \mathbb{D},$ the
representation 
\begin{equation}
h(z,t)=\int_{a}^{b}e^{yz-\frac{1}{2}y^{2}t}\mu (dy)\text{,}
\label{h-representation}
\end{equation}%
with\  \ $\int_{a}^{b}ye^{yz}\mu (dy)<\infty ,$ \  \ $a=0^{+},$ $a>0,$ $b\leq
\infty .$

We easily deduce that for each $t_{0}\geq 0,$ the function $h\left(
.,t_{0}\right) $ is absolutely monotonic, since $\partial ^{i}h\left(
z,t_{0}\right) /\partial z^{i}>0,$ $i=1,2...$ Such functions satisfy, for
each $t_{0}\geq 0,$ $i=1,2,...,$ the inequality%
\begin{equation}
\frac{\partial ^{i+1}h\left( z,t_{0}\right) }{\partial z^{i+1}}\frac{%
\partial ^{i-1}h\left( z,t_{0}\right) }{\partial z^{i-1}}-\left( \frac{%
\partial ^{i}h\left( z,t_{0}\right) }{\partial z^{i}}\right) ^{2}>0.
\label{AC}
\end{equation}

From (\ref{exp tran}), (\ref{risk-tolerance}) and (\ref{h-representation}),
we obtain that the risk tolerance function is represented as 
\begin{equation}
r(x,t)=h_{z}\left( h^{(-1)}(x,t),t\right) =\int_{a}^{b}ye^{yh^{(-1)}(x,t)-%
\frac{1}{2}y^{2}t}\mu (dy).  \label{r-h}
\end{equation}%
Furthermore, the first equality together with (\ref{h pde}) yields that it
satisfies the (ill-posed) non-linear equation 
\begin{equation}
r_{t}+\frac{1}{2}r^{2}r_{xx}=0,  \label{r-eqn}
\end{equation}%
with $r(x,0)=\int_{a}^{b}ye^{yh^{(-1)}(x,0)}\mu (dy).$

We also have that 
\begin{equation}
r_{x}(x,t)=\frac{h_{zz}\left( h^{(-1)}(x,t),t\right) }{r\left( x,t\right) }=%
\frac{1}{r\left( x,t\right) }\int_{a}^{b}y^{2}e^{yh^{(-1)}(x,t)-\frac{1}{2}%
y^{2}t}\mu (dy)>0.  \label{r-increasing}
\end{equation}%
Furthermore, 
\begin{equation}
r_{xx}\left( x,t\right) =\frac{1}{r^{3}\left( x,t\right) }\left. \left(
h_{zzz}\left( z,t\right) h_{z}\left( z,t\right) -h_{zz}\left( z,t\right)
^{2}\right) \right \vert _{z=h^{(-1)}(x,t)}>0,  \label{r-convex}
\end{equation}%
where we used (\ref{AC}).

We note that we will frequently differentiate under the integral sign in (%
\ref{h-representation}), which is permitted as explained in \cite{MZ10a}%
\textbf{. }It can be also seen directly since, after differentiation, one
can show that the relevant integrands are jointly continuous in their
respective arguments and thus uniformly locally integrable. This allows us
to differentiate under the integral sign (see, for example, Theorem 24.5 in 
\cite{Aliprantis} and the remarks following it).

As stated in the introduction, the aim herein is to investigate the spatial
and temporal limits in (\ref{limits}), with $r\left( x,t\right) $ as in (\ref%
{r-h}) when the measure has \textit{finite }support. We first provide an
example which shows that, contrary to the classical case, these two limits
do \textit{not in }general coincide.

\subsection{A motivating example}

Let the underlying measure $\mu $ be a Dirac function at $\frac{1}{1-\gamma }
$, $\gamma <1$. From (\ref{h-representation}) and (\ref{exp tran}) we have
that, for $t\geq 0,$ 
\begin{equation*}
h(x,t)=e^{\frac{1}{1-\gamma }x-\frac{1}{2}\left( \frac{1}{1-\gamma }\right)
^{2}t}\text{ \  \  \ and \  \ }u_{x}(x,t)=x^{\gamma -1}e^{-\frac{\gamma }{%
2\left( 1-\gamma \right) }t}.
\end{equation*}

Therefore, the local risk tolerance function is given by $r(x,t)=\frac{1}{%
1-\gamma }x$ and thus the spatial and temporal limits coincide,%
\begin{equation*}
\lim_{x\uparrow \infty }\frac{r(x,t_{0})}{x}=\frac{1}{1-\gamma }\text{\  \  \
\ and \  \  \ }\lim_{t\uparrow \infty }\frac{r(x_{0},t)}{x_{0}}=\frac{1}{%
1-\gamma }\text{,}
\end{equation*}%
for fixed $t_{0},x_{0}$ respectively.

Next, let the measure $\mu $ be the sum of two Dirac functions at points $a=%
\frac{1}{1-\theta }$ and $b=\frac{1}{1-\gamma }$ such that $b=2a,$ with $%
0<\theta <1$ and $\gamma <1$, i.e., 
\begin{equation}
\mu =\delta _{\frac{1}{1-\theta }}+\delta _{\frac{1}{1-\gamma }}\text{ \  \  \
\  \ with \  \  \  \ }\frac{1}{1-\gamma }=2\frac{1}{1-\theta }.
\label{measure-sum-Diracs}
\end{equation}%
Then, (\ref{h-representation}) and (\ref{exp tran}) yield that $h(x,0)=e^{%
\frac{1}{1-\theta }x}+e^{\frac{1}{1-\gamma }x},$%
\begin{equation}
u_{x}(x,0)=2^{1-\theta }\left( \sqrt{1+4x}-1\right) ^{\theta -1}\text{ \  \
and \  \  \ }u_{x}^{\left( -1\right) }(x,0)=x^{-\frac{1}{1-\theta }}+x^{-\frac{%
1}{1-\gamma }}.  \label{u-x-sumDirac}
\end{equation}%
In turn, 
\begin{equation}
\lim_{x\uparrow \infty }\frac{u_{x}(x,0)}{x^{\gamma -1}}=\lim_{x\uparrow
\infty }\frac{2^{2\left( 1-\gamma \right) }\left( \sqrt{1+4x}-1\right)
^{2\left( \gamma -1\right) }}{x^{\gamma -1}}=1.  \label{u-x-Dirac-lomit}
\end{equation}%
Moreover, expression (\ref{h-representation}) gives, for $t>0,$ 
\begin{equation*}
h(x,t)=e^{\frac{1}{1-\theta }x-\frac{1}{2}\frac{1}{\left( 1-\theta \right)
^{2}}t}+e^{\frac{2}{1-\theta }x-\frac{1}{2}\frac{2}{\left( 1-\theta \right)
^{2}}t},
\end{equation*}%
and, thus,%
\begin{equation}
h^{(-1)}(x,t)=\frac{1}{1-\theta }t+\left( 1-\theta \right) \ln \left( \frac{%
\sqrt{e^{\left( \frac{1}{1-\theta }\right) ^{2}t}+4x}-\sqrt{e^{\left( \frac{1%
}{1-\theta }\right) ^{2}t}}}{2}\right) .  \label{h-inverse}
\end{equation}%
In turn, transformation (\ref{exp tran}) yields 
\begin{equation*}
u_{x}(x,t)=2^{1-\theta }e^{(\frac{1}{2}-\frac{1}{1-\theta })t}\left( \sqrt{%
e^{\left( \frac{1}{1-\theta }\right) ^{2}t}+4x}-\sqrt{e^{\left( \frac{1}{%
1-\theta }\right) ^{2}t}}\right) ^{\gamma -1}.
\end{equation*}%
Differentiating the above to obtain $u_{xx}(x,t)$ (or using (\ref%
{h-representation}), (\ref{h-inverse}) and (\ref{r-h})), we deduce that the
risk tolerance function is given by 
\begin{equation}
r(x,t)=\frac{x}{1-\gamma }\frac{\sqrt{4x+e^{\left( \frac{1}{1-\theta }%
\right) ^{2}t}}}{\sqrt{e^{(\frac{1}{1-\theta })^{2}t}+4x}+\sqrt{e^{(\frac{1}{%
1-\theta })^{2}t}}}.  \label{r-Dirac-forward}
\end{equation}%
Therefore, for each $t_{0}\geq 0,$ 
\begin{equation}
\lim_{x\uparrow \infty }\frac{r(x,t_{0})}{x}=\frac{2}{1-\theta }=\frac{1}{%
1-\gamma }.  \label{temporal-example}
\end{equation}%
while, for each $x_{0}>0,$ 
\begin{equation}
\lim_{t\uparrow \infty }\frac{r(x_{0},t)}{x_{0}}=\frac{1}{1-\theta }.
\label{spatial-example}
\end{equation}%
Therefore, the spatial and temporal limits do \textit{not }coincide.

Next, we make the following two important observations. Firstly, note that (%
\ref{measure-sum-Diracs}) yields that the support of the measure is 
\begin{equation*}
\text{\textit{supp}}\left( \mu \right) =\left \{ \frac{1}{1-\theta },\frac{1%
}{1-\gamma }\right \} .
\end{equation*}%
Therefore, the \textit{spatial} limit coincides with the \textit{right-end}
of the support while the \textit{temporal} limit with the\textit{\ left-end }%
one.

Secondly, for each $x_{0}>0$ the temporal limit of the ratio $\frac{%
h^{(-1)}(x_{0},t)}{t}$ is equal to \textit{half }of the \textit{left-end }%
point, since (\ref{h-inverse}) yields 
\begin{equation*}
\lim_{t\uparrow \infty }\frac{h^{(-1)}(x_{0},t)}{t}
\end{equation*}%
\begin{equation*}
=\lim_{t\uparrow \infty }\left( \frac{1}{1-\theta }+\frac{1-\theta }{t}\ln
\left( \frac{1}{2}\left( \sqrt{e^{(\frac{1}{1-\theta })^{2}t}+4x}-\sqrt{e^{(%
\frac{1}{1-\theta })^{2}t}}\right) \right) \right) =\frac{1}{2\left(
1-\theta \right) }.
\end{equation*}%
In section 4 we will show that these two properties are always valid. In
particular, we will see that it is the limit of the above ratio that plays
the key role in establishing the temporal turnpike limit for general
measures.

\bigskip

To juxtapose the above results with the ones in the traditional expected
terminal utility setting, we compute the analogous quantities and associated
limits for the cases analyzed in \cite{HZ99} and \cite{B-Zheng} for
log-normal markets. Without loss of generality, we consider a market with a
riskless bond of zero interest rate and a single log-normal stock with mean
rate of return $\mu $ and volatility $\sigma .$

To this end, we fix an arbitrary horizon $T>0$ and, in analogy to (\ref%
{u-x-sumDirac}), we take the \textit{terminal} inverse marginal utility, $%
I\left( x\right) =\left( U^{\prime }\right) ^{\left( -1\right) }\left(
x\right) ,$ to be%
\begin{equation*}
I\left( x\right) =x^{-\frac{1}{1-\theta }}+x^{-\frac{1}{1-\gamma }},
\end{equation*}%
for $x>0$ and $\theta ,\gamma $ as in (\ref{measure-sum-Diracs}). This
corresponds to terminal marginal utility $U^{\prime }\left( x\right) =\left( 
\frac{\sqrt{1+4x}-1}{2}\right) ^{\theta -1}$ and, thus, in analogy to (\ref%
{u-x-Dirac-lomit}),%
\begin{equation*}
\lim_{x\uparrow \infty }\frac{U^{\prime }\left( x\right) }{x^{\gamma -1}}=1.
\end{equation*}

We now consider the value function, say $u\left( x,t;T\right) $ of the
associated Merton problem, for $t\in \left[ 0,T\right] .$ Letting $\tau =T-t$
be the time to the end of the investment horizon, we deduce, using well
known results, that the function $\tilde{u}\left( x,\tau \right) \equiv
u\left( x,T-t;T\right) ,$ satisfies, for $(x,\tau )\in \mathbb{R}_{+}\times %
\left[ 0,T\right) $, the Hamilton-Jacobi-Bellman equation 
\begin{equation*}
\tilde{u}{_{\tau }+}\frac{1}{2}\lambda ^{2}{\frac{\tilde{u}_{x}^{2}}{\tilde{u%
}_{xx}}=0}.
\end{equation*}%
The inverse spatial marginal value function $\tilde{v}:\mathbb{R}_{+}\times %
\left[ 0,T\right) \rightarrow \mathbb{R}_{+}$ then solves%
\begin{equation*}
\tilde{v}{_{\tau }=}\frac{1}{2}\lambda ^{2}x{^{2}\tilde{v}_{xx}}+\lambda
^{2}x\tilde{v}_{x},
\end{equation*}%
with $\tilde{v}(x,0)=I\left( x\right) .$ We easily deduce that 
\begin{equation*}
\tilde{v}(x,\tau )=e^{\alpha \tau }x^{-\alpha }+e^{\beta \tau }x^{-2\alpha },
\end{equation*}%
with ${\alpha =}\frac{1}{2}\lambda ^{2}\frac{\theta }{\left( 1-\theta
\right) ^{2}}$ and ${\beta =\lambda }^{2}\frac{1+\theta }{\left( 1-\theta
\right) ^{2}}{.}$ Note that $\beta >2\alpha $.

Taking the spatial inverse of $\tilde{v}(x,\tau )$ yields 
\begin{equation*}
\tilde{u}_{x}\left( x,\tau \right) =\left( \frac{e^{\alpha \tau }+\sqrt{%
e^{2\alpha \tau }+4xe^{\beta \tau }}}{2x}\right) ^{1-\theta }.
\end{equation*}%
Therefore, the associated risk tolerance function is given by%
\begin{equation*}
\tilde{r}(x,\tau )=\frac{1}{1-\theta }\left( \frac{2x}{1+\sqrt{1+4xe^{(\beta
-2\alpha )\tau }}}+\frac{8x^{2}}{\left( {\sqrt{e^{(2\alpha -\beta )\tau }}+%
\sqrt{e^{(2\alpha -\beta )\tau }+4x}}\right) ^{2}}\right) .
\end{equation*}%
In turn, for each $\tau _{0}>0$ and $x_{0}>0,$ we obtain, respectively, the
spatial and the temporal limits, 
\begin{equation*}
\lim_{x\uparrow \infty }\frac{\tilde{r}(x,\tau _{0})}{x}=\frac{1}{1-\theta }%
\text{ \  \  \  \  \ and \  \  \  \ }\lim_{\tau \uparrow \infty }\frac{\tilde{r}%
(x_{0},\tau )}{x_{0}}=\frac{1}{1-\theta }.
\end{equation*}

\section{Spatial asymptotic results}

We examine the spatial asymptotic behavior of the risk tolerance function as 
$x\uparrow \infty ,$ for each $t_{0}\geq 0$, under asymptotic assumptions
for large wealth levels of the investor's initial risk preferences. In
accordance with similar assumptions in \cite{HZ99} and \cite{B-Zheng}, we
impose this asymptotic assumption on the marginal $u_{0}^{\prime }\left(
x\right) $ instead of the function itself.

\textbf{Assumption 1:} \textit{The initial datum }$u_{0}$ satisfies, for
some $\gamma <1,$ 
\begin{equation}
\lim_{x\uparrow \infty }\frac{u_{0}^{\prime }(x)}{x^{\gamma -1}}=1.
\label{asymptotic1}
\end{equation}

The next result yields necessary and sufficient conditions on $b,$ the right
end of the support of the measure, for the above assumption to hold.

\begin{lemma}
Assumption (\ref{asymptotic1}) holds if and only if the associated measure $%
\mu $ satisfies 
\begin{equation}
b=\frac{1}{1-\gamma }\text{\  \  \  \ and \  \ }\mu \left( \left \{ \frac{1}{%
1-\gamma }\right \} \right) =1.  \label{consequence1}
\end{equation}
\end{lemma}

\begin{proof}
From (\ref{asymptotic1}), (\ref{exp tran})\ and the fact that $h(x,0)$ is
strictly increasing and of full range, we have 
\begin{equation}
1=\lim_{x\uparrow \infty }\frac{u_{x}(x,0)}{x^{\gamma -1}}=\lim_{z\uparrow
\infty }\frac{u_{x}(h(z,0),0)}{(h(z,0))^{\gamma -1}}=\lim_{z\uparrow \infty
}\left( \frac{h(z,0)}{e^{\frac{1}{1-\gamma }z}}\right) ^{1-\gamma }.
\label{u-x-h}
\end{equation}%
Therefore, representation (\ref{h-representation}) gives 
\begin{equation}
\lim_{z\uparrow \infty }\int_{a}^{b}e^{z(y-\frac{1}{1-\gamma })}\mu (dy)=1.
\label{limit of h(x,0)}
\end{equation}%
If $a=b,$ then (\ref{consequence1})\ follows directly. If $a<b,$ then, it
must be that $a\leq \frac{1}{1-\gamma },$ otherwise, we get a contradiction.
In turn, for $\varepsilon >0,$ 
\begin{equation}
\int_{a}^{b}e^{z(y-\frac{1}{1-\gamma })}\mu (dy)\geq \int_{\frac{1}{1-\gamma 
}+\varepsilon }^{b}e^{z(y-\frac{1}{1-\gamma })}\mu (dy)\geq e^{\varepsilon
z}\mu \left( \lbrack \frac{1}{1-\gamma }+\varepsilon ,b]\right) .
\label{ineq of h(x,0)}
\end{equation}%
Sending $\varepsilon \downarrow 0$ and using (\ref{limit of h(x,0)}) yield
that $\mu ((\frac{1}{1-\gamma },b])=0,$ and thus, supp$(\mu )\subseteq (a,%
\frac{1}{1-\gamma }].$ Moreover, we have from (\ref{limit of h(x,0)}) that 
\begin{equation*}
1=\lim_{z\uparrow \infty }\int_{a}^{\left( \frac{1}{1-\gamma }\right)
^{-}}e^{z(y-\frac{1}{1-\gamma })}\mu (dy)+\mu (\{ \frac{1}{1-\gamma }\})=\mu
(\{ \frac{1}{1-\gamma }\}),
\end{equation*}%
and we conclude. The rest of the proof follows easily.
\end{proof}

We next state the main spatial asymptotic result.

\begin{proposition}
Suppose that the initial datum $u_{0}$ satisfies the asymptotic property (%
\ref{asymptotic1}). Then, for each $t_{0}\geq 0$, the relative risk
tolerance converges to the right-end of the support of the measure $\mu ,$ 
\begin{equation}
\lim_{x\uparrow \infty }\frac{r(x,t_{0})}{x}=\frac{1}{1-\gamma }.
\label{x-lim of r(x,t)}
\end{equation}
\end{proposition}

\begin{proof}
Let $t_{0}\geq 0$. From representation (\ref{ineq of h(x,0)}) we have that 
\begin{equation*}
h\left( z,t_{0}\right) =\int_{a}^{\left( \frac{1}{1-\gamma }\right)
^{-}}e^{zy-\frac{1}{2}t_{0}y^{2}}\mu (dy)+e^{\frac{1}{1-\gamma }z-\frac{1}{2}%
\left( \frac{1}{1-\gamma }\right) ^{2}t_{0}},
\end{equation*}%
and, in turn, the dominated convergence theorem implies 
\begin{equation}
\lim_{z\uparrow \infty }\frac{h(z,t_{0})}{e^{\frac{1}{1-\gamma }z-\frac{1}{2}%
\left( \frac{1}{1-\gamma }\right) ^{2}t_{0}}}=1.  \label{limit of h(x,t)}
\end{equation}%
Therefore, from (\ref{exp tran}), together with the strict monotonicity and
full range of $h(z,t_{0})$, we deduce that 
\begin{equation}
\lim_{x\uparrow \infty }\frac{u_{x}(x,t_{0})}{x^{\gamma -1}e^{-\frac{\gamma 
}{2(1-\gamma )}t_{0}}}=1,  \label{u-x-limit}
\end{equation}%
since 
\begin{equation*}
\lim_{x\uparrow \infty }\frac{u_{x}(x,t_{0})}{x^{\gamma -1}e^{-\frac{\gamma 
}{2(1-\gamma )}t_{0}}}=\lim_{z\uparrow \infty }\frac{e^{-z+\frac{t_{0}}{2}}}{%
h^{\gamma -1}(z,t_{0})e^{-\frac{\gamma }{2(1-\gamma )}t_{0}}}
\end{equation*}%
\begin{equation*}
=\lim_{z\uparrow \infty }\left( \frac{h(z,t_{0})}{e^{\frac{1}{1-\gamma }z-%
\frac{1}{2}\left( \frac{1}{1-\gamma }\right) ^{2}t_{0}}}\right) ^{1-\gamma
}=1.
\end{equation*}%
Next, we claim that 
\begin{equation}
\lim_{x\uparrow \infty }\frac{u_{xx}(x,t_{0})}{x^{\gamma -2}e^{-\frac{\gamma 
}{2(1-\gamma )}t_{0}}}=\frac{1}{\gamma -1}.  \label{u-xx-limit}
\end{equation}%
To prove this, it suffices to show that for any $t_{0}\geq 0$, $%
u_{x}(x,t_{0})$ is convex since the above would then follow from the
arguments in Lemma 3.1 (ii) in \cite{HZ99}. To this end, differentiating (%
\ref{exp tran}) yields 
\begin{equation}
u_{xxx}\left( h(z,t_{0}),t_{0}\right) \left( h_{z}(z,t_{0})\right)
^{2}+u_{xx}(h(z,t_{0}),t_{0})h_{zz}(z,t_{0})=e^{-z+\frac{t_{0}}{2}}.
\label{twice diff}
\end{equation}%
The strict convexity of $h$ and the strict concavity of $u$ then give%
\begin{equation}
u_{xxx}\left( h(z,t_{0}),t_{0}\right) >0,  \label{prudence-aux}
\end{equation}%
and using the strict monotonicity and full range of $h$ we conclude.

Combining (\ref{u-x-limit}) and (\ref{u-xx-limit}) yields 
\begin{equation*}
\lim_{x\uparrow \infty }\frac{r(x,t_{0})}{x}=\lim_{x\uparrow \infty }\left( -%
\frac{u_{x}(x,t_{0})}{xu_{xx}(x,t_{0})}\right)
\end{equation*}%
\begin{equation*}
=\lim_{x\uparrow \infty }\left( -\frac{u_{x}(x,t_{0})}{x^{\gamma -1}e^{-%
\frac{\gamma }{2(1-\gamma )}t_{0}}}\left( \frac{u_{xx}(x,t_{0})}{x^{\gamma
-2}e^{-\frac{\gamma }{2(1-\gamma )}t_{0}}}\right) ^{-1}\right) =\frac{1}{%
1-\gamma }.
\end{equation*}
\end{proof}

\bigskip

We stress that assumption (\ref{asymptotic1}), or equivalently (\ref%
{consequence1}), \textit{cannot} be weakened. Indeed, as we will see in
example 6.2,\textbf{\ }$\ $where we take the measure to be the Lebesgue on $%
\left[ a,b\right] ,$ and thus there is no mass at $b,$ the spatial turnpike
property does not hold.

\begin{corollary}
Suppose that the initial datum $u_{0}$ satisfies the asymptotic property (%
\ref{asymptotic1}). Then, for each $t_{0}\geq 0,$%
\begin{equation}
\lim_{x\uparrow \infty }r_{x}\left( x,t_{0}\right) =\frac{1}{1-\gamma }.
\label{x-limit-derivative}
\end{equation}
\end{corollary}

\begin{proof}
From (\ref{r-convex}) we have that, for each $t_{0}\geq 0,$ $\lim_{x\uparrow
\infty }r_{x}\left( x,t_{0}\right) $ exists, and we easily conclude.
\end{proof}

\section{Temporal (turnpike) asymptotic results}

We investigate the temporal asymptotic behavior of the relative risk
tolerance as $t\uparrow \infty ,$ for each $x_{0}>0,$ under asymptotic
assumption of the initial marginal utility for large wealth levels. This is
the genuine "turnpike" analogue of similar results in classical expected
utility models and the main finding herein. It shows that the relative risk
tolerance will converge to the \textit{left-end }of the support of the
underlying measure $\mu $.

As in the spatial case, we first relate the properties of the measure to the
asymptotic behavior of the initial (marginal) datum.

\textbf{Assumption 2:}\textit{\ There exists }$\gamma <1,$\textit{\ }$\gamma
\neq 0$, \textit{such that for all }$\gamma ^{\prime }\in \left( \gamma
,1\right) ,$%
\begin{equation}
\lim_{x\uparrow \infty }\frac{u_{0}^{\prime }\left( x\right) }{x^{\gamma
^{\prime }-1}}=0,  \label{upper-gamma}
\end{equation}%
\textit{while, for all }$\gamma ^{^{\prime \prime }}<\gamma ,$%
\begin{equation}
\lim_{x\uparrow \infty }\frac{u_{0}^{\prime }\left( x\right) }{x^{\gamma
^{^{\prime \prime }}-1}}=\infty .  \label{lower-gamma}
\end{equation}

As we show next, the above assumption is directly related to a condition
introduced in \cite{Huberman-Ross} and \cite{DRB99}, for a discrete and a
continuous-time case, respectively.

\begin{lemma}
Assumption 2 is equivalent to the function $u_{0}^{\prime }\left( x\right) $
varying regularly at infinity with exponent $\gamma -1$, i.e. for all $k>0,$%
\begin{equation}
\lim_{x\uparrow \infty }\frac{u_{0}^{\prime }(kx)}{u_{0}^{\prime }(x)}%
=k^{\gamma -1}.  \label{H-R}
\end{equation}
\end{lemma}

\begin{proof}
We first show that condition (\ref{H-R}) implies (\ref{upper-gamma}) and (%
\ref{lower-gamma}). We argue by contradiction. Suppose that (\ref%
{upper-gamma}) does not hold. Then, there exists $\gamma ^{\prime }\in
(\gamma ,1)$ and $\varepsilon >0$ such that for $x$ large enough, $\frac{%
u_{0}^{\prime }\left( x\right) }{x^{\gamma ^{\prime }-1}}>\varepsilon .$ On
the other hand, condition (\ref{H-R}) implies that, for all $k>0$ and $x$
large enough, $\left \vert \frac{u_{0}^{\prime }(kx)}{u_{0}^{\prime }\left(
x\right) k^{\gamma -1}}-1\right \vert <\varepsilon .$ Thus, for large enough 
$x$,

\begin{equation*}
0<\frac{u_{0}^{\prime }(kx)}{(kx)^{\gamma \prime -1}}=\frac{u_{0}^{\prime
}(kx)}{u_{0}^{\prime }(x)k^{\gamma -1}}\frac{u_{0}^{\prime }(x)}{x^{\gamma
^{\prime }-1}}k^{\gamma -\gamma ^{\prime }}<(1+\varepsilon )\frac{%
u_{0}^{\prime }(x)}{x^{\gamma ^{\prime }-1}}k^{\gamma -\gamma ^{\prime }}.
\end{equation*}%
Since $\gamma -\gamma ^{\prime }<0$, $\lim_{k\uparrow \infty }\frac{%
u_{0}^{\prime }(kx)}{(kx)^{\gamma ^{\prime }-1}}=0,$ and we conclude.
Working similarly, we establish (\ref{lower-gamma}).

Next, we show the reverse direction. Assume that (\ref{lower-gamma}) and (%
\ref{upper-gamma}) hold. Then, for all $\delta ,k>0$ and $x$ large enough, 
\begin{equation*}
\frac{u_{0}^{\prime }(kx)}{(kx)^{\gamma +\delta -1}}<1~~~\  \  \  \text{and \  \ 
}~~\frac{x^{\gamma -\delta -1}}{u_{0}^{\prime }(x)}<1.
\end{equation*}%
Multiplying these two equations and rearranging gives, for all $\delta >0,$%
\begin{equation*}
\frac{u_{0}^{\prime }(kx)}{u_{0}^{\prime }(x)}<\frac{(kx)^{\gamma +\delta -1}%
}{x^{\gamma -\delta -1}}=k^{\gamma +\delta -1}x^{2\delta }.
\end{equation*}%
Similarly, it follows from interchanging $kx$ and $x$ in the above two
inequalities that%
\begin{equation*}
\frac{u_{0}^{\prime }(kx)}{u_{0}^{\prime }(x)}>\frac{(kx)^{\gamma -\delta -1}%
}{x^{\gamma +\delta -1}}=k^{\gamma -\delta -1}x^{-2\delta },
\end{equation*}%
and condition (\ref{H-R}) follows by sending first $\delta \downarrow 0$ and
then $x\uparrow \infty .$
\end{proof}

Assumption 2 is weaker than Assumption 1, and implies, as we show next, that
the measure $\mu $ has support with right-end point at $\frac{1}{1-\gamma },$
but without necessarily having a mass therein. \  \ 

\begin{lemma}
Assumption 2 holds if and only if the measure $\mu $ has finite support with
its right boundary at $\frac{1}{1-\gamma },$ namely,%
\begin{equation}
\inf \left \{ y>0:\mu \left( \left( y,\infty \right) \right) =0\right \} =%
\frac{1}{1-\gamma }.  \label{right-support}
\end{equation}
\end{lemma}

\begin{proof}
We show that Assumption 2 implies property (\ref{right-support}).\textit{\ }%
For each $\gamma ^{\prime }\in \left( \gamma ,1\right) ,$ we deduce from (%
\ref{upper-gamma}) that 
\begin{equation*}
0=\lim_{x\uparrow \infty }\frac{u_{x}(x,0)}{x^{\gamma ^{\prime }-1}}%
=\lim_{z\uparrow \infty }\frac{u_{x}(h(z,0),0)}{(h(z,0))^{^{\gamma ^{\prime
}-1}}}=\lim_{z\uparrow \infty }\left( \frac{h\left( z,0\right) }{e^{\frac{z}{%
1-\gamma ^{\prime }}}}\right) ^{1-\gamma ^{\prime }},
\end{equation*}%
and, thus, 
\begin{equation}
\lim_{z\uparrow \infty }\int_{a}^{b}e^{z\left( y-\frac{1}{1-\gamma ^{\prime }%
}\right) }\mu \left( dy\right) =0.  \label{zero-limit}
\end{equation}%
Next, observe that if $b\geq 1,$ then it will contradict the above limit,
and thus we need to have $b<1.$ Assume now that there exists $\gamma
^{\prime }\in \left( \gamma ,1\right) $ with $b=\frac{1}{1-\gamma ^{\prime }}%
.$ Then, for each $\tilde{\gamma}\in \left( \gamma ,\gamma ^{\prime }\right) 
$ we have $\frac{1}{1-\tilde{\gamma}}<\frac{1}{1-\gamma ^{\prime }}$ and the
above gives, for $\varepsilon $ small enough,

\begin{equation*}
\lim_{z\uparrow \infty }\left( \int_{a}^{\left( \frac{1}{1-\tilde{\gamma}}%
+\varepsilon \right) ^{-}}e^{z\left( y-\frac{1}{1-\tilde{\gamma}}\right)
}\mu \left( dy\right) +\int_{\frac{1}{1-\tilde{\gamma}}+\varepsilon
}^{b}e^{z\left( y-\frac{1}{1-\tilde{\gamma}}\right) }\mu \left( dy\right)
\right) =0.
\end{equation*}%
Therefore, it must be that $\mu \left( \lbrack \frac{1}{1-\tilde{\gamma}}%
+\varepsilon ,b]\right) =0.$ Sending $\varepsilon \downarrow 0,$ gives $\mu
\left( \left( \frac{1}{1-\tilde{\gamma}},b\right] \right) =0,$ which is a
contradiction. Thus, we must have $b\leq \frac{1}{1-\gamma }.$ Similarly,
using (\ref{lower-gamma}) we obtain that $b\geq \frac{1}{1-\gamma },$ and,
thus, $b=\frac{1}{1-\gamma }.$

To show the reverse direction, we first observe that property (\ref%
{right-support}) and the dominated convergence theorem yield that, for any $%
\varepsilon >0,$

\begin{equation*}
\lim_{z\uparrow \infty }h(z,0)e^{-(\frac{1}{1-\gamma }+\varepsilon
)z}=\lim_{z\uparrow \infty }\int_{a}^{\frac{1}{1-\gamma }}e^{z\left(
y-\left( \frac{1}{1-\gamma }+\varepsilon \right) \right) }\mu (dy)=0.
\end{equation*}%
Then, setting $\gamma ^{\prime }$ such that $\frac{1}{1-\gamma ^{\prime }}=%
\frac{1}{1-\gamma }+\varepsilon ,$ we deduce (\ref{upper-gamma}) for all $%
\gamma ^{\prime }\in \left( \gamma ,1\right) $.

The rest of the proof follows easily and it is thus omitted.
\end{proof}

\bigskip

We have so far established that under Assumption 2 the associated measure $%
\mu $ has a finite right boundary (but not necessarily a mass) at $\frac{1}{%
1-\gamma }$, \ and vice-versa.

We now turn our attention to the left boundary of the support, denoted by $%
a, $ where 
\begin{equation}
a:=\inf \{y\geq 0:\mu \left( (0,y]\right) >0\}.  \label{a-point}
\end{equation}%
In the upcoming proofs we will frequently use the identity%
\begin{equation}
x_{0}=\int_{a}^{\frac{1}{1-\gamma }}e^{yh^{(-1)}(x_{0},t)-\frac{1}{2}%
y^{2}t}\mu (dy),  \label{x-h}
\end{equation}%
for $x_{0}>0,$ which follows directly from (\ref{h-representation}) for $b=%
\frac{1}{1-\gamma }.$

\begin{lemma}
Let $h^{\left( -1\right) }:\mathbb{D}_{+}\rightarrow \mathbb{R}$ be the
spatial inverse of $h,$ and $a$ as in (\ref{a-point}). Then, for each $%
x_{0}>0,$ $\lim_{t\uparrow \infty }h_{t}^{(-1)}(x_{0},t)$ exists and,
moreover, for $t\geq 0,$ 
\begin{equation}
\frac{a}{2}\leq h_{t}^{(-1)}(x_{0},t)\leq \frac{1}{2\left( 1-\gamma \right) }%
.  \label{h_inv deriv}
\end{equation}
\end{lemma}

\begin{proof}
Let $x_{0}>0$ and observe that (\ref{h pde}) yields 
\begin{equation*}
h_{t}^{\left( -1\right) }\left( x_{0},t\right) =\frac{1}{2}\frac{%
h_{xx}\left( h^{\left( -1\right) }\left( x_{0},t\right) \right) }{%
h_{x}\left( h^{\left( -1\right) }\left( x_{0},t\right) ,t\right) }=\frac{1}{2%
}\frac{{\int_{a}^{\frac{1}{1-\gamma }}y^{2}e^{yh^{(-1)}(x_{0},t)-\frac{1}{2}%
y^{2}t}\mu (dy)}}{{\int_{a}^{\frac{1}{1-\gamma }}ye^{yh^{(-1)}(x_{0},t)-%
\frac{1}{2}y^{2}t}\mu (dy)}}
\end{equation*}%
and thus inequality (\ref{h_inv deriv}) holds, for all $t\geq 0.$

To show that $\lim_{t\uparrow \infty }h_{t}^{(-1)}(x_{0},t)$ exists, it
suffices to show that $h_{t}^{(-1)}(x_{0},t)$ is decreasing in time. Indeed,
direct calculations yield 
\begin{equation}
h_{tt}^{(-1)}(x_{0},t)=-\frac{{\int_{a}^{\frac{1}{1-\gamma }}\left(
yh_{t}^{(-1)}(x_{0},t)-\frac{1}{2}y^{2}\right) ^{2}e^{yh^{(-1)}(x_{0},t)-%
\frac{1}{2}y^{2}t}\mu (dy)}}{{\int_{a}^{\frac{1}{1-\gamma }%
}ye^{yh^{(-1)}(x_{0},t)-\frac{1}{2}y^{2}t}\mu (dy)}}<0.  \label{h_inv 2deriv}
\end{equation}%
Alternatively, differentiating $h\left( h^{\left( -1\right) }\left(
x_{0},t\right) ,t\right) =x_{0}$ twice yields, setting $z=h^{\left(
-1\right) }\left( x_{0},t\right) ,$%
\begin{equation*}
h_{tt}^{(-1)}(x_{0},t)h_{x}\left( z,t\right) +\left( h_{t}^{\left( -1\right)
}\left( x_{0},t\right) \right) ^{2}h_{xx}\left( z,t\right) +2h_{t}^{\left(
-1\right) }\left( x_{0},t\right) h_{xt}\left( z,t\right) +h_{tt}\left(
z,t\right) =0.
\end{equation*}%
We have that both $h_{x},$ $h_{xx}>0,$ as it follows directly from (\ref%
{h-representation}) and differentiation. Furthermore, the above quadratic in 
$h_{t}^{\left( -1\right) }\left( x,t\right) $ remains positive, which would
then yield that $h_{tt}^{(-1)}(x_{0},t)<0.$ Indeed, 
\begin{equation*}
h_{xt}^{2}\left( z,t\right) -h_{xx}\left( z,t\right) h_{tt}\left( z,t\right)
=h_{xxx}^{2}\left( z,t\right) -h_{xx}\left( z,t\right) h_{xxxx}\left(
z,t\right) <0,
\end{equation*}%
as it follows from (\ref{AC}).
\end{proof}

\bigskip

We are now ready to present one of the main findings herein, which yields
the limit as $t\uparrow \infty $ of the ratio $\frac{1}{t}h^{(-1)}(x_{0},t).$
We show that it converges to half of the lower-end of the measure's support.
Some related weaker results can be found in \cite{RS73}.

\begin{proposition}
Let $h^{\left( -1\right) }:\mathbb{D}_{+}\rightarrow \mathbb{R}$ be the
spatial inverse of the function $h$ (cf. (\ref{h-representation})) and let $%
a,b$ be the left and right end of the support, respectively, with $a=0^{+}$
or $a>0,$ and $b<\infty .$ Then, for each $x_{0}>0,$ the following
assertions hold.

i) It holds that 
\begin{equation}
\lim_{t\uparrow \infty }\frac{h^{(-1)}(x_{0},t)}{t}=\frac{a}{2}.
\label{limit-a/2}
\end{equation}

ii) Let 
\begin{equation}
\Delta \left( x_{0},t\right) :=\frac{h^{(-1)}(x_{0},t)}{t}-\frac{a}{2}.
\label{delta-difference}
\end{equation}%
If $a>0,$ then 
\begin{equation}
\left \vert \Delta \left( x_{0},t\right) \right \vert \leq \frac{1}{at}\ln
\left( \frac{\mu \left( \lbrack a,\frac{1}{1-\gamma }]\right) }{x_{0}}%
\right) ,\text{ \ if \ }\Delta \left( x_{0},t\right) <0,  \label{delta<0}
\end{equation}%
and 
\begin{equation}
x_{0}\geq \mu \left( \left[ a,a+\Delta \left( x_{0},t\right) \right] \right)
e^{\frac{1}{2}ta\Delta \left( x_{0},t\right) },\text{ \ if \ }\Delta \left(
x_{0},t\right) >0.  \label{delta>0}
\end{equation}%
If $a=0^{+},$ then $\Delta \left( x_{0},t\right) >0,$ and, moreover, for
each $\theta \in \left( 0,1\right) ,$ 
\begin{equation}
x_{0}\geq \mu \left( \left[ \Delta \left( x_{0},t\right) ,\left( 1+\theta
\right) \Delta \left( x_{0},t\right) \right] \right) e^{\frac{1}{2}t\left(
1-\theta ^{2}\right) \Delta ^{2}\left( x_{0},t\right) }.
\label{delta-zero-bound}
\end{equation}
\end{proposition}

\begin{proof}
\textit{i).} Let $x_{0}>0$ fixed. Recall that $h_{t}^{(-1)}(x_{0},t)>0$ (cf.
(\ref{h_inv deriv})) and, thus, $\lim_{t\uparrow \infty }h^{(-1)}(x_{0},t)$
exists. Moreover, rewriting (\ref{x-h})\ as 
\begin{equation}
x_{0}=\int_{a}^{\frac{1}{1-\gamma }}e^{ty\left( \frac{h^{(-1)}(x_{0},t)}{t}-%
\frac{1}{2}y\right) }\mu (dy),  \label{h-inverse-formula}
\end{equation}%
we see that $\lim_{t\uparrow \infty }h^{(-1)}(x_{0},t)=\infty ,$ otherwise,
sending $t\uparrow \infty $ we get a contradiction. In turn, from Lemma 7
and L' Hospital's rule, we deduce that

\begin{equation}
A(x_{0}):=\lim_{t\uparrow \infty }\frac{h^{(-1)}(x_{0},t)}{t}%
=\lim_{t\uparrow \infty }h_{t}^{(-1)}(x_{0},t),  \label{h_inv/t=h_inv_t}
\end{equation}%
and thus 
\begin{equation}
\frac{a}{2}\leq A(x_{0})\leq \frac{1}{2(1-\gamma )}.
\label{bounds on limit of h_inv/t}
\end{equation}

Next, we claim that $A\left( x_{0}\right) <\frac{1}{2\left( 1-\gamma \right) 
}.$

Let $a>0.$ If $a=\frac{1}{1-\gamma },$ then $a=b$ and $h^{\left( -1\right)
}\left( x_{0},t\right) =\ln x_{0}^{1-\gamma }+\frac{1}{2}\frac{1}{1-\gamma }%
t,$ and the result follows directly.

Let $0<a<\frac{1}{1-\gamma }.$ Assume that there exists $x_{0}$ such that $%
A\left( x_{0}\right) =\frac{1}{2\left( 1-\gamma \right) }.$ Then, for $%
\varepsilon >0,$ there exists $t_{0}\left( x_{0},\varepsilon \right) $ such
that, for $t\geq t_{0},$ 
\begin{equation*}
-\varepsilon \leq \frac{h^{(-1)}(x_{0},t)}{t}-\frac{1}{2(1-\gamma )}\leq
\varepsilon .
\end{equation*}%
In turn, for $\delta >0$ small enough, the above and (\ref{x-h})\ yield 
\begin{equation*}
x_{0}\geq \int_{a}^{\left( \frac{1}{1-\gamma }-2\varepsilon -\delta \right)
^{-}}e^{ty\left( \frac{1}{2(1-\gamma )}-\varepsilon -\frac{1}{2}y\right)
}\mu (dy)+\int_{\frac{1}{1-\gamma }-2\varepsilon -\delta }^{\frac{1}{%
1-\gamma }}e^{ty(\frac{1}{2(1-\gamma )}-\varepsilon -\frac{1}{2}y)}\mu (dy),
\end{equation*}%
which yields a contradiction as $t\uparrow \infty ,$ because the first
integral would converge to $\infty .$

Next, assume that there exists $x_{0}>0$ such that 
\begin{equation}
\frac{a}{2}<A(x_{0})<\frac{1}{2(1-\gamma )}.  \label{contradiction}
\end{equation}%
Then, for $\varepsilon ,\delta >0$ small enough we have 
\begin{equation}
a<2(A(x_{0})-\varepsilon )-\delta <2(A(x_{0})-\varepsilon )<\frac{1}{%
1-\gamma }.  \label{inequality-aux}
\end{equation}%
From (\ref{x-h}), we then deduce that, for $t\geq t_{0}(x_{0},\varepsilon )$%
, $x_{0}\geq \int_{a}^{\frac{1}{1-\gamma }}e^{t(y(A(x_{0})-\varepsilon )-%
\frac{1}{2}y^{2})}\mu (dy).$ If $\mu \left( \left \{ a\right \} \right) \neq
0, $ then $x_{0}\geq e^{\frac{ta}{2}\left( 2\left( A\left( x_{0}\right)
-\varepsilon \right) -a\right) }\mu \left( \left \{ a\right \} \right) ,$
and sending $t\uparrow \infty $ yields a contradiction. If $\mu \left(
\left
\{ a\right \} \right) =0,$ then 
\begin{equation}
x_{0}\geq \int_{a}^{\frac{1}{1-\gamma }}e^{t(y(A(x_{0})-\varepsilon )-\frac{1%
}{2}y^{2})}\mu (dy)\geq \int_{a}^{2(A(x_{0})-\varepsilon )-\delta
}e^{t(y(A(x_{0})-\varepsilon )-\frac{1}{2}y^{2})}\mu (dy).  \label{quadratic}
\end{equation}%
Consider the quadratic $B\left( y\right) =y(A(x_{0})-\varepsilon )-\frac{1}{2%
}y^{2}.$ We have 
\begin{equation*}
B\left( y_{1}\right) =B\left( y_{2}\right) =0,\text{ for \ }y_{1}=0\text{ \
and \ }y_{2}=2\left( A(x_{0})-\varepsilon \right) ,
\end{equation*}%
$B\left( y\right) >0,$ for $0<y<2\left( A(x_{0})-\varepsilon \right) ,$ and $%
B\left( y\right) $ achieves a maximum at $y^{\ast }=A(x_{0})-\varepsilon .$

Next, we look at its minimum, $y_{\ast }=\min_{a\leq y\leq
2(A(x)-\varepsilon )-\delta }\Delta \left( y\right) ,$ and claim that 
\begin{equation}
y_{\ast }=2(A(x_{0})-\varepsilon )-\delta .  \label{minimum}
\end{equation}%
Indeed, if $0<a\leq y^{\ast },$ then choosing $\delta $ $<a,$ direct
calculations yield $\Delta \left( a\right) >\Delta \left( y_{\ast }\right) .$
If $y^{\ast }<a,$ then (\ref{inequality-aux}) yields $a<y_{\ast }<y_{2},$
and, thus, the minimum also occurs at $y_{\ast }.$

Clearly, because $y_{1}<y_{\ast }<y_{2},$ we have $B\left( y_{\ast }\right) =%
\frac{1}{2}\delta \left( 2(A(x_{0})-\varepsilon )-\delta \right) >0.$
Therefore, for $t\geq t_{0}(x_{0},\varepsilon )$, 
\begin{equation}
x_{0}\geq \int_{a}^{2(A(x_{0})-\varepsilon )-\delta }e^{tB\left( y_{\ast
}\right) }\mu (dy).  \label{contrad}
\end{equation}%
As $t\uparrow \infty $, the right hand side of (\ref{contrad}) converges to $%
\infty $, unless it holds that $\mu \left( \lbrack a,2(A(x_{0})-\varepsilon
)-\delta ]\right) =0.$ Sending $\delta \downarrow 0$ and $\varepsilon
\downarrow 0$, we then have 
\begin{equation*}
\mu ([a,2A(x_{0})])=0,
\end{equation*}%
which, however, contradicts (\ref{contradiction}). Therefore, it must be
that that, for all $x>0$, $A(x_{0})\leq \frac{a}{2}$, and we easily conclude.

If $a=0^{+},$ similar arguments yield that for every $\theta \in \left(
0,A\left( x_{0}\right) \right] ,$ we have that $\mu (\left[ \theta ,2A(x_{0})%
\right] )=0.$ Sending $\theta \downarrow 0$ yields $\mu \left( 0,2A\left(
x_{0}\right) \right] =0,$ which contradicts (\ref{contradiction}).

\textit{ii). }Let $a>0$.

If $\Delta \left( x_{0},t\right) <0,$ from (\ref{x-h}) we have%
\begin{equation*}
x_{0}=\int_{a}^{\frac{1}{1-\gamma }}e^{ty\left( \Delta \left( x_{0},t\right)
+\frac{1}{2}\left( a-y\right) \right) }\mu (dy)
\end{equation*}%
\begin{equation*}
\leq e^{ta\Delta \left( x_{0},t\right) }\int_{a}^{\frac{1}{1-\gamma }}e^{%
\frac{1}{2}ty\left( a-y\right) }\mu \left( dy\right) \leq e^{ta\Delta \left(
x_{0},t\right) }\mu \left( \left[ a,\frac{1}{1-\gamma }\right] \right) ,
\end{equation*}%
and (\ref{delta<0}) follows.

If $\Delta \left( x_{0},t\right) >0,$ then (\ref{limit-a/2}) yields that,
for $\varepsilon $ small enough and $t\geq t_{0}\left( x_{0},\varepsilon
\right) ,$ $0<\frac{h^{^{\left( -1\right) }}\left( x_{0},t\right) }{t}-\frac{%
a}{2}<\varepsilon .$ Choosing $\varepsilon $ such that $\varepsilon <\frac{1%
}{2\left( 1-\gamma \right) }-\frac{a}{2}$ yields $0<\frac{h^{^{\left(
-1\right) }}\left( x_{0},t\right) }{t}-\frac{a}{2}<\frac{1}{2\left( 1-\gamma
\right) }-\frac{a}{2},$ and using that $a<\frac{1}{1-\gamma },$ gives%
\begin{equation*}
\frac{a}{2}+\frac{h^{^{\left( -1\right) }}\left( x_{0},t\right) }{t}\leq 
\frac{1}{1-\gamma }.
\end{equation*}%
From (\ref{h-inverse}) we then deduce that%
\begin{equation*}
x_{0}\geq \int_{a}^{\frac{a}{2}+\frac{h\left( -1\right) \left(
x_{0},t\right) }{t}}e^{ty\left( \frac{h^{\left( -1\right) }\left(
x_{0},t\right) }{t}-\frac{y}{2}\right) }\mu \left( dy\right) .
\end{equation*}

The quadratic $H\left( y\right) :=y\left( \frac{h^{\left( -1\right) }\left(
x_{0},t\right) }{t}-\frac{y}{2}\right) $ in the above integrand becomes zero
at $y_{1}=0$ and $y_{3}=2\frac{h^{^{\left( -1\right) }}\left( x_{0},t\right) 
}{t}>a$ and, therefore, its minimum occurs at one of the end points $a$ or $%
\frac{a}{2}+\frac{h^{^{\left( -1\right) }}\left( x_{0},t\right) }{t}.$ Note
that $a<\frac{a}{2}+\frac{h^{^{\left( -1\right) }}\left( x_{0},t\right) }{t}%
<y_{3}$.

If it occurs at $a,$ then $H\left( a\right) =a\Delta \left( x_{0},t\right) ,$
while if it occurs at $\frac{a}{2}+\frac{h^{^{\left( -1\right) }}\left(
x_{0},t\right) }{t},$ then $H\left( \frac{a}{2}+\frac{h^{^{\left( -1\right)
}}\left( x_{0},t\right) }{t}\right) =$ $\frac{1}{2}\left( \frac{a}{2}+\frac{%
h^{^{\left( -1\right) }}\left( x_{0},t\right) }{t}\right) \Delta \left(
x_{0},t\right) >\frac{1}{2}a\Delta \left( x_{0},t\right) .$

Combining the above gives 
\begin{equation*}
x_{0}\geq \int_{a}^{\frac{a}{2}+\frac{h^{^{\left( -1\right) }}\left(
x_{0},t\right) }{t}}e^{\frac{1}{2}ta\Delta \left( x_{0},t\right) }\mu \left(
dy\right) =\mu \left( \left[ a,a+\Delta \left( x_{0},t\right) \right]
\right) e^{\frac{1}{2}ta\Delta \left( x_{0},t\right) }.
\end{equation*}

Finally, let $a=0^{+}.$ Then, $\Delta \left( x_{0},t\right) =\frac{h^{\left(
-1\right) }\left( x_{0},t\right) }{t}.$

Recall that $\lim_{t\uparrow \infty }h^{^{\left( -1\right) }}\left(
x_{0},t\right) =\infty ,$ and thus $\frac{h^{^{\left( -1\right) }}\left(
x_{0},t\right) }{t}>0,$ for $t$ large. For $\varepsilon \in \left( \frac{%
h^{^{\left( -1\right) }\left( x_{0},t\right) }}{t},2\frac{h^{^{\left(
-1\right) }\left( x_{0},t\right) }}{t}\right) $ we then have%
\begin{equation*}
x_{0}\geq \int_{\frac{h^{\left( -1\right) }\left( x_{0},t\right) }{t}%
}^{\varepsilon }e^{ty\left( \frac{h^{\left( -1\right) }\left( x_{0},t\right) 
}{t}-\frac{y}{2}\right) }\mu \left( dy\right) \geq \int_{\frac{h^{^{\left(
-1\right) }\left( x_{0},t\right) }}{t}}^{\varepsilon }e^{t\varepsilon \left( 
\frac{h^{\left( -1\right) }\left( x_{0},t\right) }{t}-\frac{\varepsilon }{2}%
\right) }\mu \left( dy\right) .
\end{equation*}%
Setting $\varepsilon =\left( 1+\theta \right) \frac{h^{^{\left( -1\right)
}}\left( x_{0},t\right) }{t},$ (\ref{delta-zero-bound})\ follows.
\end{proof}

\bigskip

We are now ready to prove one of the main results herein.

\begin{theorem}
Let $a$ be the left end of the support of the measure $\mu $. Then, for each 
$x_{0}>0,$ 
\begin{equation}
\lim_{t\uparrow \infty }\frac{r\left( x_{0},t\right) }{x_{0}}=a.
\label{temporal-limit}
\end{equation}%
Furthermore, there exists a function\textbf{\ }$G\left( x_{0},t\right) $
given by 
\begin{equation*}
G\left( x_{0},t\right) :=\left \{ 
\begin{array}{c}
\int_{a}^{\frac{1}{1-\gamma }}(y-a)e^{-ty(\frac{y-a}{2})}\mu (dy),\text{ \ }%
\Delta \left( x_{0},t\right) <0\text{\ } \\ 
\\ 
2\Delta \left( x_{0},t\right) x_{0}+\int_{a+2\Delta \left( x_{0},t\right) }^{%
\frac{1}{1-\gamma }}(y-a)e^{ty(\frac{2\Delta \left( x_{0},t\right) +a-y}{2}%
)}\mu (dy),\text{ }\Delta \left( x_{0},t\right) >0,%
\end{array}%
\right.
\end{equation*}%
satisfying with $\lim_{t\uparrow \infty }G\left( x_{0},t\right) =0$ and, for 
$t$ large enough$,$%
\begin{equation}
0\leq r(x_{0},t)-ax_{0}\leq G\left( x_{0},t\right) .  \label{G-limit}
\end{equation}
\end{theorem}

\begin{proof}
We present two alternative convergence proofs. The first yields (\ref%
{temporal-limit}) while the second gives the rate of convergence $G\left(
x_{0},t\right) .$

To this end, differentiating (\ref{exp tran}) gives 
\begin{equation}
u_{xt}(x_{0},t)=\left( \frac{1}{2}-h_{t}^{(-1)}(x_{0},t)\right)
u_{x}(x_{0},t).  \label{u_xt eq}
\end{equation}%
Moreover, (\ref{u pde}) and (\ref{risk-tolerance}) imply that $%
u_{t}(x_{0},t)=-\frac{1}{2}u_{x}(x_{0},t)r(x_{0},t)$ and, in turn, 
\begin{equation}
u_{tx}\left( x_{0},t\right) =-\frac{1}{2}u_{xx}\left( x_{0},t\right) r\left(
x_{0},t\right) -\frac{1}{2}u_{x}\left( x_{0},t\right) r_{x}\left(
x_{0},t\right) .  \label{u_tx eq}
\end{equation}%
Combining the above we deduce 
\begin{equation}
\frac{1}{2}r_{x}(x_{0},t)=h_{t}^{(-1)}(x_{0},t),  \label{r_x=h_inv_t}
\end{equation}%
and from Proposition 8 and (\ref{h_inv/t=h_inv_t})%
\begin{equation}
\lim_{t\uparrow \infty }r_{x}(x_{0},t)=\lim_{t\uparrow \infty
}2h_{t}^{(-1)}(x_{0},t)=a.  \label{t-limit of r_x}
\end{equation}%
On the other hand, 
\begin{equation*}
\lim_{c\downarrow 0^{+}}\int_{c}^{x_{0}}r_{x}(\rho ,t)d\rho
=r(x_{0},t)-\lim_{c\downarrow 0^{+}}r(c,t).
\end{equation*}%
Using the fact that, for all $t\geq 0,$ $\lim_{x\downarrow 0^{+}}r(x,t)=0$
(see \cite{MZ10a}), we get that, for $x_{0}>0$, 
\begin{equation}
r(x_{0},t)=\int_{a}^{x_{0}}r_{x}(\rho ,t)d\rho .
\label{fundamental thm of calculus}
\end{equation}

Finally, we deduce from (\ref{r_x=h_inv_t}) and (\ref{h_inv 2deriv}) that $%
r_{xt}\left( x_{0},t\right) <0,$ and thus, for $x_{0}>0$, we have for $y\in
(0,x_{0}]$, $r_{x}(y,t)\leq r_{x}(x_{0},0).$ However, for all $x_{0}>0,$ $%
r_{x}(x_{0},0)<\infty .$ This follows directly from (\ref{r-h}),(\ref%
{h-representation}) and the full range of $h\left( x,0\right) ,$ since 
\begin{equation*}
r_{x}\left( h\left( z,0\right) ,0\right) =\frac{h_{zz}\left( z,0\right) }{%
h_{z}\left( z,0\right) }=\frac{\int_{a}^{\frac{1}{1-\gamma }}y^{2}e^{yz-%
\frac{1}{2}t^{2}y}\mu \left( dy\right) }{\int_{a}^{\frac{1}{1-\gamma }%
}ye^{yz-\frac{1}{2}t^{2}y}\mu \left( dy\right) }\leq \frac{1}{1-\gamma }.
\end{equation*}%
Using the dominated convergence theorem and passing to the limit as $%
t\uparrow \infty $ in (\ref{r_x=h_inv_t}), we deduce (\ref{temporal-limit}).

Next, we give the second convergence proof, which also yields the rate of
convergence. First note that%
\begin{equation}
0\leq r(x_{0},t)-ax_{0}.  \label{inequality}
\end{equation}%
This follows directly from (\ref{r-h}), (\ref{h-representation}) and (\ref%
{x-h}), for 
\begin{equation*}
r\left( x_{0},t\right) =\int_{a}^{\frac{1}{1-\gamma }}ye^{t(y\frac{%
h^{(-1)}(x_{0},t)}{t}-\frac{1}{2}y^{2})}\mu (dy)\geq a\int_{a}^{\frac{1}{%
1-\gamma }}e^{t(y\frac{h^{(-1)}(x_{0},t)}{t}-\frac{1}{2}y^{2})}\mu (dy).
\end{equation*}%
Furthermore, from (\ref{r-h}), (\ref{h-representation}), (\ref{x-h}) and (%
\ref{delta-difference}), we have 
\begin{equation}
r(x_{0},t)-ax_{0}=\int_{a}^{\frac{1}{1-\gamma }}(y-a)e^{ty(\frac{2\Delta
\left( x_{0},t\right) +a-y}{2})}\mu (dy).  \label{r-integral}
\end{equation}

If $\Delta \left( x_{0},t\right) <0$ (which occurs only if $a>0,$ as shown
in the previous proof), then the above yields 
\begin{equation*}
r(x_{0},t)-ax_{0}\leq \int_{a}^{\frac{1}{1-\gamma }}(y-a)e^{-ty(\frac{y-a}{2}%
)}\mu (dy),
\end{equation*}%
and (\ref{G-limit}) follows directly with $G\left( t\right) :=\int_{a}^{%
\frac{1}{1-\gamma }}(y-a)e^{-ty(\frac{y-a}{2})}\mu (dy).$

Let $\Delta \left( x_{0},t\right) >0$ and $a>0$ or $a=0^{+}.$ If $a=\frac{1}{%
1-\gamma },$ then the result follows trivially.

For $a<\frac{1}{1-\gamma },$ observe that for $t$ large enough, $0<a+2\Delta
\left( x_{0},t\right) <\frac{1}{1-\gamma },$ and thus representation (\ref%
{r-integral}) gives%
\begin{equation*}
r\left( x_{0},t\right) -ax_{0}=\int_{a}^{\left( a+2\Delta \left(
x_{0},t\right) \right) ^{-}}(y-a)e^{ty(\frac{2\Delta \left( x_{0},t\right)
+a-y}{2})}\mu (dy)
\end{equation*}%
\begin{equation*}
+\int_{a+2\Delta \left( x_{0},t\right) }^{\frac{1}{1-\gamma }}(y-a)e^{ty(%
\frac{2\Delta \left( x_{0},t\right) +a-y}{2})}\mu (dy).
\end{equation*}%
Let $C_{1}\left( x_{0},t\right) :=\int_{a}^{\left( a+2\Delta \left(
x_{0},t\right) \right) ^{-}}(y-a)e^{ty(\frac{2\Delta \left( x_{0},t\right)
+a-y}{2})}\mu (dy),$ and observe that 
\begin{equation*}
C_{1}\left( x_{0},t\right) \leq 2\Delta \left( x_{0},t\right)
\int_{a}^{\left( a+2\Delta \left( x_{0},t\right) \right) ^{-}}e^{ty(\frac{%
2\Delta \left( x_{0},t\right) +a-y}{2})}\mu (dy)\leq 2\Delta \left(
x_{0},t\right) x_{0},
\end{equation*}%
where we used (\ref{x-h}). Thus%
\begin{equation}
\lim_{t\uparrow \infty }C_{1}\left( x_{0},t\right) =0.  \label{A}
\end{equation}%
Let also $C_{2}\left( x_{0},t\right) :=\int_{a+2\Delta \left( x_{0},t\right)
}^{\frac{1}{1-\gamma }}(y-a)e^{ty(\frac{2\Delta \left( x_{0},t\right) +a-y}{2%
})}\mu (dy)$ and $F\left( y,t,x_{0}\right) :=\left( y-a\right) e^{ty(\frac{%
2\Delta \left( x_{0},t\right) +a-y}{2})},$ $y\in \left[ a+2\Delta \left(
x_{0},t\right) ,\frac{1}{1-\gamma }\right] .$ Then, $F\left( a+2\Delta
\left( x_{0},t\right) ,t,x_{0}\right) =2\Delta \left( x_{0},t\right) ,$ and
thus $\lim_{t\uparrow \infty }F\left( a+2\Delta \left( x_{0},t\right)
,t,x_{0}\right) =0.$ Furthermore, for each $y\in \left( a+2\Delta \left(
x_{0},t\right) ,\frac{1}{1-\gamma }\right] ,$ we also have $\lim_{t\uparrow
\infty }F\left( y,t,x_{0}\right) =0.$ In turn, the dominated convergence
theorem gives 
\begin{equation}
\lim_{t\uparrow \infty }C_{2}\left( x_{0},t\right) =0.  \label{B}
\end{equation}%
Setting $G\left( x_{0},t\right) :=C_{1}\left( x_{0},t\right) +C_{2}\left(
x_{0},t\right) ,$ and using (\ref{A}) and (\ref{B}), we obtain (\ref{G-limit}%
).
\end{proof}

\section{Spatial and temporal limits for the relative prudence function}

\bigskip We now revert our attention to the relative prudence function $%
p\left( x,t\right) $ defined, for $\left( x,t\right) \in \mathbb{D}_{+},$ as 
\begin{equation}
p\left( x,t\right) =-\frac{xu_{xxx}\left( x,t\right) }{u_{xx}\left(
x,t\right) },  \label{prudence}
\end{equation}%
with $u$ solving (\ref{u pde}).

\begin{proposition}
\textit{For }$\left( x,t\right) \in \mathbb{D}_{+},$ we have that $p\left(
x,t\right) >0.$ Moreover, the following spatial and temporal limits hold.

\textit{i) If Assumption 1 holds, then, for each }$t_{0}\geq 0,$%
\begin{equation}
\lim_{x\uparrow \infty }p(x,t_{0})=2-\gamma .  \label{prudence-space}
\end{equation}

\textit{ii)\ If Assumption 2 holds, then, for each }$x_{0}>0,$%
\begin{equation}
\lim_{t\uparrow \infty }p(x_{0},t)=\left \{ 
\begin{array}{c}
1+\frac{1}{a},\text{ \ if \ }a>0 \\ 
\\ 
\infty ,\text{ if }a=0^{+}.%
\end{array}%
\right.  \label{prudence-time}
\end{equation}
\end{proposition}

\begin{proof}
Using (\ref{prudence}) and (\ref{risk-tolerance}), we deduce that, for each $%
t_{0}\geq 0,$ 
\begin{equation*}
p\left( x,t_{0}\right) =\frac{x}{r\left( x,t_{0}\right) }\left(
1+r_{x}\left( x,t_{0}\right) \right) ,
\end{equation*}%
and the fact that $p\left( x,t_{0}\right) >0$ and (\ref{prudence-space})
follow directly from (\ref{r-increasing})\ and (\ref{x-lim of r(x,t)}),
respectively.

From (\ref{prudence}) and equation (\ref{u pde}) we also obtain that, for
each $x_{0}>0,$ 
\begin{equation}
\frac{u_{xt}\left( x_{0},t\right) }{u_{x}\left( x_{0},t\right) }=1-\frac{1}{2%
}\frac{r\left( x_{0},t\right) }{x_{0}}p\left( x_{0},t\right) =\frac{1}{2}%
-h_{t}^{(-1)}(x_{0},t).  \label{prudence-0}
\end{equation}%
Using that $\lim_{t\uparrow \infty }h_{t}^{(-1)}(x_{0},t)=\frac{a}{2}$ we
easily conclude.
\end{proof}

\section{Examples}

We present two representative examples in which the measure is,
respectively, a sum of Dirac functions and the Lebesgue measure. The first
example generalizes the results of the example in subsection 2.1, while the
second demonstrates that the spatial turnpike property fails if there is no
mass at the right end of the measure's support.

\subsection{Finite sum of Dirac functions}

We assume that 
\begin{equation*}
\mu =\dsum \limits_{n=1}^{N}\delta _{y_{n}},\text{ \  \ with \  \ }%
0<y_{1}<\dots <y_{N}=\frac{1}{1-\gamma }.
\end{equation*}%
Then, $h(z,0)=\Sigma _{n=1}^{N}e^{y_{n}z}$ and, thus, $\lim_{z\uparrow
\infty }h\left( z,0\right) e^{-zy_{N}}=1.$ In turn, (\ref{u-x-h})\ yields 
\begin{equation*}
\lim_{x\uparrow \infty }\frac{u_{x}\left( x,0\right) }{x^{\gamma -1}}=1,
\end{equation*}%
which verifies the results of Lemma 2. We also have, for $\left( z,t\right)
\in \mathbb{D},$ 
\begin{equation*}
h(z,t)=\sum_{n=1}^{N}\exp \left( y_{n}z-\frac{1}{2}y_{n}^{2}t\right) .
\end{equation*}%
(cf. (\ref{h-representation})), and, therefore, for $x>0,$ 
\begin{equation}
x=\sum_{n=1}^{N}\exp \left( y_{n}t\left( \frac{h^{(-1)}(x,t)}{t}-\frac{1}{2}%
y_{n}\right) \right) .  \label{h-inverse-Dirac}
\end{equation}%
Furthermore, 
\begin{equation}
h^{(-1)}(x,t)-\frac{1}{2}y_{1}t\leq \frac{1}{y_{1}}\ln x.  \label{extra}
\end{equation}

\subsubsection{Temporal asymptotic expansion of $h^{(-1)}(x_{0},t)$ for
large $t$}

We claim that, for each $x_{0}>0,$ as $t\uparrow \infty $, 
\begin{equation}
h^{(-1)}(x_{0},t)=\frac{1}{2}y_{1}t+\frac{1}{y_{1}}\ln x_{0}+o(1).
\label{Dirac-temporal}
\end{equation}%
Indeed, using the limit (\ref{limit-a/2}), we have 
\begin{equation*}
\lim_{t\uparrow \infty }\left( \frac{h^{(-1)}(x_{0},t)}{t}-\frac{1}{2}%
y_{n}\right) 
\begin{cases}
<0,\text{ \  \ } & \text{ }1<n\leq N \\ 
=0,\text{ \  \  \ } & \text{ }n=1.%
\end{cases}%
\end{equation*}%
Therefore, as $t\uparrow \infty $, all the terms in (\ref{h-inverse-Dirac})
vanish except for the first one, and thus, 
\begin{equation}
x_{0}=\lim_{t\uparrow \infty }\exp \left( y_{1}h^{(-1)}(x_{0},t)-\frac{1}{2}%
y_{1}^{2}t\right) .  \label{vanish}
\end{equation}%
Taking logarithm and rearranging terms yields (\ref{Dirac-temporal}).

\subsubsection{Spatial asymptotic expansion of $h^{(-1)}(x,t_{0})$ for large 
$x$}

We claim that, for each $t_{0}\geq 0,$ 
\begin{equation}
h^{(-1)}(x,t_{0})=\left( 1-\gamma \right) \ln x+\frac{1}{2\left( 1-\gamma
\right) }t_{0}+o(1).  \label{Dh_inv x-expansion}
\end{equation}%
To obtain this, we first establish that%
\begin{equation}
\lim_{x\uparrow \infty }\frac{h^{(-1)}(x,t_{0})}{\ln x}=\left( 1-\gamma
\right) .  \label{Dh_inv x-lim}
\end{equation}%
Indeed, fix $t_{0}\geq 0,$ let $\delta \in (0,\frac{1}{1-\gamma })$ and
assume that 
\begin{equation*}
\liminf_{x\uparrow \infty }\frac{h^{(-1)}(x,t_{0})}{\ln x}<\frac{1}{\frac{1}{%
1-\gamma }+\delta }.
\end{equation*}%
Then, using (\ref{h-inverse-Dirac}) and that $h^{(-1)}(x,t_{0})>0,$ for
large $x$, we have 
\begin{equation*}
1=\frac{1}{x}\sum_{n=1}^{N}\exp \left( y_{n}\ln x\left( \frac{%
h^{(-1)}(x,t_{0})}{\ln x}\right) -\frac{1}{2}y_{n}^{2}t_{0}\right)
\end{equation*}%
\begin{equation*}
\leq \frac{1}{x}\sum_{n=1}^{N}\exp \left( y_{n}\ln x\left( \frac{%
h^{(-1)}(x,t_{0})}{\ln x}\right) \right) \leq Nx^{\frac{1}{1-\gamma }\frac{%
h^{(-1)}(x,t_{0})}{\ln x}-1},
\end{equation*}%
and using that $\frac{1}{1-\gamma }\frac{1}{\frac{1}{1-\gamma }+\delta }-1=-%
\frac{\delta \left( 1-\gamma \right) }{1+\delta \left( 1-\gamma \right) }<0,$
we get a contradiction as $x\uparrow \infty $.

Since $\delta $ is arbitrary, we deduce that 
\begin{equation}
\liminf_{x\uparrow \infty }\frac{h^{(-1)}(x,t_{0})}{\ln x}\geq \left(
1-\gamma \right) .  \label{Dh_inv liminf}
\end{equation}%
Similarly, assume that for $\delta \in \left( 0,\frac{1}{1-\gamma }\right) ,$
\begin{equation*}
\limsup_{x\uparrow \infty }\frac{h^{(-1)}(x,t_{0})}{\ln x}>\frac{1}{\frac{1}{%
1-\gamma }-\delta }.
\end{equation*}%
Then, (\ref{extra})\ gives%
\begin{equation*}
1>\frac{1}{x}\exp \left( \frac{1}{1-\gamma }\ln x\frac{h^{(-1)}(x,t_{0})}{%
\ln x}-\frac{1}{2}\left( \frac{1}{1-\gamma }\right) ^{2}t_{0}\right)
\end{equation*}%
\begin{equation*}
=x^{\frac{1}{1-\gamma }\frac{h^{(-1)}(x,t_{0})}{\ln x}-1}e^{-\frac{1}{2}%
\left( \frac{1}{1-\gamma }\right) ^{2}t_{0}}
\end{equation*}%
and using that $\frac{1}{1-\gamma }\frac{1}{\frac{1}{1-\gamma }-\delta }-1=%
\frac{\delta \left( 1-\gamma \right) }{1-\delta \left( 1-\gamma \right) }>0,$
we get a contradiction as $x\uparrow \infty .$ Since $\delta $ is arbitrary,
we deduce that 
\begin{equation}
\limsup_{x\uparrow \infty }\frac{h^{(-1)}(x,t_{0})}{\ln x}\leq \left(
1-\gamma \right) ,  \label{Dh_inv limsup}
\end{equation}%
and we easily conclude.

Next, we rewrite (\ref{h-inverse-Dirac}) as%
\begin{equation}
1=\sum_{n=1}^{N}\exp \left( y_{n}h^{(-1)}(x,t_{0})-\frac{1}{2}%
y_{n}^{2}t_{0}-\ln x\right)  \label{rew}
\end{equation}%
\begin{equation*}
=\sum_{n=1}^{N}\exp \left( y_{n}\ln x\left( \frac{h^{(-1)}(x,t_{0})}{\ln x}-%
\frac{1}{y_{n}}\right) -\frac{1}{2}y_{n}^{2}t_{0}\right) .
\end{equation*}%
Note that from the limit in (\ref{Dh_inv x-lim}) we have that 
\begin{equation*}
\lim_{x\uparrow \infty }\left( \frac{h^{(-1)}(x,t_{0})}{\ln x}-\frac{1}{y_{n}%
}\right) =%
\begin{cases}
<0, & \text{ }1\leq n<N \\ 
=0, & \text{ }n=N.%
\end{cases}%
\end{equation*}%
Therefore, as $x\uparrow \infty $, the first $N-1$ terms in (\ref{rew})
vanish, and we deduce that 
\begin{equation*}
\lim_{x\uparrow \infty }\exp \left( \frac{1}{1-\gamma }h^{(-1)}(x,t_{0})-\ln
x-\frac{1}{2}\left( \frac{1}{1-\gamma }\right) ^{2}t_{0}\right) =1.
\end{equation*}%
We then obtain (\ref{Dh_inv x-expansion}) by taking the logarithm and
rearranging the terms.

\subsubsection{Spatial and temporal asymptotics of $r(x,t)$}

From representation (\ref{r-h}), we have for the risk tolerance function 
\begin{equation}
r(x,t)=\sum_{n=1}^{N}y_{n}\exp \left( y_{n}h^{(-1)}(x,t)-\frac{1}{2}%
y_{n}^{2}t\right) .  \label{Dr}
\end{equation}%
Let $x_{0}>0$. Then, (\ref{extra}) gives%
\begin{equation*}
r(x_{0},t)\leq \sum_{n=1}^{N}y_{n}\exp \left( y_{n}(\frac{1}{2}y_{1}t+\frac{1%
}{y_{1}}\ln x_{0})-\frac{1}{2}y_{n}^{2}t\right)
\end{equation*}%
\begin{equation*}
=y_{1}x_{0}+\sum_{n=2}^{N}y_{n}\exp \left( \frac{1}{2}y_{n}(y_{1}-y_{n})t%
\right) x_{0}^{\frac{y_{n}}{y_{1}}}.
\end{equation*}%
Therefore, the temporal asymptotic expansion of $r(x_{0},t)$ as $t\uparrow
\infty $ is given by 
\begin{equation}
r(x_{0},t)=y_{1}x_{0}+O\left( e^{\frac{1}{2}y_{2}(y_{1}-y_{2})t}\right) .
\label{r-temporal-Dirac}
\end{equation}%
Next, let $t_{0}\geq 0$. Then,%
\begin{equation*}
\lim_{x\uparrow \infty }r(x,t_{0})=\lim_{x\uparrow \infty
}\sum_{n=1}^{N}y_{n}\exp \left( y_{n}(\left( 1-\gamma \right) \ln x+\frac{1}{%
2\left( 1-\gamma \right) }t_{0})-\frac{1}{2}y_{n}^{2}t_{0}\right) ,
\end{equation*}%
and, thus, as $x\uparrow \infty ,$ 
\begin{equation}
r(x,t_{0})=\sum_{n=1}^{N}y_{n}\exp \left( \frac{1}{2}y_{n}t_{0}(\frac{1}{%
1-\gamma }-y_{n})\right) x^{\left( 1-\gamma \right) y_{n}}+o(1).
\label{r-spatial-Dirac}
\end{equation}%
Therefore, for each $x_{0}>0$ and $t_{0}\geq 0,$ we have the temporal
asymptotic expansion (\ref{r-temporal-Dirac}) yields 
\begin{equation*}
\lim_{t\uparrow \infty }\frac{r(x_{0},t)}{x_{0}}=y_{1}\text{ \  \ and \  \ }%
\lim_{x\uparrow \infty }\frac{r\left( x,t_{0}\right) }{x}=y_{N}=\frac{1}{%
1-\gamma },
\end{equation*}%
and these limits are consistent with the findings in Proposition 3 and
Theorem 9, respectively.

\subsection{Lebesgue measure}

We consider a case of a measure with continuous support but without a mass
at its right boundary. We derive the associated limits and also show that
the spatial turnpike property fails.

\begin{itemize}
\item Lebesgue measure on $\left[ a,\frac{1}{1-\gamma }\right] ,$ $a>0$
\end{itemize}

Consider the functions $\varphi (z):=e^{-\frac{z^{2}}{2}}~~$and$~~\Phi
(z):=\int_{-\infty }^{z}\varphi (y)dy,$ for $z\in \mathbb{R}.$ Then,
representations (\ref{h-representation})\ and (\ref{x-h}) yield,
respectively, 
\begin{equation}
h(z,t)=\int_{a}^{\frac{1}{1-\gamma }}e^{yz-\frac{1}{2}y^{2}t}dy=\frac{%
e^{z^{2}/2t}}{\sqrt{t}}\int_{a\sqrt{t}-z/\sqrt{t}}^{\frac{1}{1-\gamma }\sqrt{%
t}-z/\sqrt{t}}\varphi (y)dy,  \label{h1}
\end{equation}%
and 
\begin{equation}
x=\int_{a}^{\frac{1}{1-\gamma }}e^{yt\left( \frac{h^{(-1)}(x,t)}{t}-\frac{1}{%
2}y\right) }dy=\frac{1}{\sqrt{t}}e^{\frac{h^{(-1)}(x,t)^{2}}{2t}}\int_{a%
\sqrt{t}-\frac{h^{(-1)}(x,t)}{\sqrt{t}}}^{\frac{1}{1-\gamma }\sqrt{t}-\frac{%
h^{(-1)}(x,t)}{\sqrt{t}}}\varphi (y)dy.  \label{h_inv 2}
\end{equation}

\subsubsection{Temporal asymptotic expansion of $h^{(-1)}(x_{0},t)$ for
large $t$}

We claim that for $x_{0}>0,$ as $t\uparrow \infty ,$ 
\begin{equation}
h^{(-1)}(x_{0},t)=\frac{1}{2}at+\frac{1}{a}\left( \ln t+\ln x_{0}+\ln \frac{a%
}{2}\right) +o(1).  \label{t-expansion1}
\end{equation}%
To show this, we first establish that 
\begin{equation}
x_{0}=\lim_{t\uparrow \infty }\frac{e^{a(h^{(-1)}(x_{0},t)-\frac{1}{2}at)}}{%
\frac{1}{2}at}.  \label{intermediate t-limit}
\end{equation}%
Using (\ref{h_inv 2}) and that, for $z<0,$ 
\begin{equation}
\Phi (z)\leq -\frac{\varphi (z)}{z},  \label{aux1}
\end{equation}%
we have, for $t$ large enough, 
\begin{equation*}
x_{0}\leq \frac{1}{\sqrt{t}}\exp \left( \frac{h^{(-1)}(x_{0},t)^{2}}{2t}%
\right) \Phi \left( -a\sqrt{t}+\frac{h^{(-1)}(x_{0},t)}{\sqrt{t}}\right)
\end{equation*}%
\begin{equation*}
\leq \frac{1}{\sqrt{t}}\frac{1}{a\sqrt{t}-\frac{h^{(-1)}(x_{0},t)}{\sqrt{t}}}%
\exp \left( \frac{h^{(-1)}(x_{0},t)^{2}}{2t}\right) \varphi \left( -a\sqrt{t}%
+\frac{h^{(-1)}(x_{0},t)}{\sqrt{t}}\right)
\end{equation*}%
\begin{equation*}
=\frac{e^{a(h^{(-1)}(x_{0},t)-\frac{1}{2}at)}}{at-h^{(-1)}(x_{0},t)}.
\end{equation*}%
In turn, 
\begin{equation}
x_{0}\leq \liminf_{t\uparrow \infty }\frac{e^{a(h^{(-1)}(x_{0},t)-\frac{1}{2}%
at)}}{at-h^{(-1)}(x_{0},t)}.  \label{inf2}
\end{equation}%
Next, we show that 
\begin{equation*}
x_{0}\geq \limsup_{t\uparrow \infty }\frac{e^{a(h^{(-1)}(x_{0},t)-\frac{1}{2}%
at)}}{at-h^{(-1)}(x_{0},t)},
\end{equation*}%
which with (\ref{inf2})\ will yield (\ref{intermediate t-limit}). To this
end, we use that for any $b>a>0,$ the inequality 
\begin{equation*}
\Phi (b)-\Phi (a)\geq \frac{1}{b}\left( \varphi (a)-\varphi (b)\right) \ 
\end{equation*}%
holds. Let $1<k<\frac{1}{a(1-\gamma )}$. From (\ref{h_inv 2}) and the above,
we have, for $t$ large enough, that

\begin{equation*}
x_{0}\geq \frac{1}{\sqrt{t}}e^{\frac{h^{(-1)}(x_{0},t)^{2}}{2t}}\left( \Phi
\left( ka\sqrt{t}-\frac{h^{(-1)}(x_{0},t)}{\sqrt{t}}\right) -\Phi \left( a%
\sqrt{t}-\frac{h^{(-1)}(x_{0},t)}{\sqrt{t}}\right) \right)
\end{equation*}%
\begin{equation*}
\geq \frac{1}{\sqrt{t}}\frac{1}{ka\sqrt{t}-\frac{h^{(-1)}(x_{0},t)}{\sqrt{t}}%
}e^{\frac{h^{(-1)}(x_{0},t)^{2}}{2t}}
\end{equation*}%
\begin{equation*}
\times \left( \varphi \left( a\sqrt{t}-\frac{h^{(-1)}(x_{0},t)}{\sqrt{t}}%
\right) -\varphi \left( ka\sqrt{t}-\frac{h^{(-1)}(x_{0},t)}{\sqrt{t}}\right)
\right)
\end{equation*}%
\begin{equation*}
=\frac{1}{kat-h^{(-1)}(x_{0},t)}\left( e^{a(h^{(-1)}(x_{0},t)-\frac{1}{2}%
at)}-e^{ka(h^{(-1)}(x_{0},t)-\frac{1}{2}kat)}\right) .
\end{equation*}%
From Proposition 8 and since $k>1$, we have 
\begin{equation*}
\lim_{t\uparrow \infty }\frac{e^{ka(h^{(-1)}(x_{0},t)-\frac{1}{2}kat)}}{%
kat-h^{(-1)}(x_{0},t)}=\lim_{t\uparrow \infty }\frac{{e^{ka^{2}t(\frac{%
h^{(-1)}(x_{0},t)}{at}-\frac{k}{2})}}}{{at\left( k-\frac{h^{(-1)}(x_{0},t)}{%
at}\right) }}=0.
\end{equation*}%
Therefore, 
\begin{equation*}
x_{0}\geq \limsup_{t\uparrow \infty }\frac{1}{kat-h^{(-1)}(x_{0},t)}\left(
e^{a(h^{(-1)}(x_{0},t)-\frac{1}{2}at)}-e^{ka(h^{(-1)}(x_{0},t)-\frac{1}{2}%
kat)}\right)
\end{equation*}%
\begin{equation*}
\geq \limsup_{t\uparrow \infty }\frac{e^{ka(h^{(-1)}(x_{0},t)-\frac{1}{2}%
kat)}}{kat-h^{(-1)}(x_{0},t)}-\lim_{t\uparrow \infty }\frac{%
e^{ka(h^{(-1)}(x_{0},t)-\frac{1}{2}kat)}}{kat-h^{(-1)}(x_{0},t)}%
=\limsup_{t\uparrow \infty }\frac{e^{a(h^{(-1)}(x_{0},t)-\frac{1}{2}at)}}{%
kat-h^{(-1)}(x_{0},t)},
\end{equation*}%
and sending $k\downarrow 1$ we conclude.

Next, we utilize the Lambert-W function $W(x)$, defined as the inverse
function of $F(x)=xe^{x}$, to derive the explicit asymptotic expansion of $%
h^{(-1)}(x_{0},t)$ as $t\uparrow \infty $. Recalling the notation $\Delta
\left( x_{0},t\right) =h^{(-1)}(x_{0},t)-\frac{1}{2}at$, we deduce from %
\eqref{intermediate t-limit} that there exists $\varepsilon (t)$ with $%
\lim_{t\uparrow \infty }\varepsilon (t)=0$, such that 
\begin{equation*}
\frac{e^{a\Delta \left( x_{0},t\right) }}{\frac{1}{2}at-\Delta \left(
x_{0},t\right) }=x_{0}(1+\varepsilon (t)).
\end{equation*}%
Rewriting it yields 
\begin{equation*}
a\left( \frac{1}{2}at-\Delta \left( x_{0},t\right) \right) e^{a(\frac{1}{2}%
at-\Delta \left( x_{0},t\right) )}=\frac{a}{x_{0}(1+\varepsilon (t))}e^{%
\frac{1}{2}a^{2}t},
\end{equation*}%
Using that the left hand side is of the form $F(a(\frac{1}{2}at-\Delta
\left( x_{0},t\right) ))$, we obtain 
\begin{equation*}
a(\frac{1}{2}at-\Delta \left( x_{0},t\right) )=W\left( \frac{a}{%
x_{0}(1+\varepsilon (t))}e^{\frac{1}{2}a^{2}t}\right) ,
\end{equation*}%
and, in turn, 
\begin{equation*}
\Delta \left( x_{0},t\right) =\frac{1}{2}at-\frac{1}{a}W\left( \frac{a}{%
x_{0}(1+\varepsilon (t))}e^{\frac{1}{2}a^{2}t}\right) .
\end{equation*}%
It is established in \cite{corless} that the asymptotic expansion of $W(x),$
for large $x,$ is given by 
\begin{equation*}
W(x)=\ln x-\ln (\ln x)+o(1).
\end{equation*}%
Therefore, 
\begin{equation*}
\Delta \left( x_{0},t\right) =\frac{1}{2}at-\frac{1}{a}\ln \left( \frac{a}{%
x_{0}(1+\varepsilon (t))}e^{\frac{1}{2}a^{2}t}\right) +\frac{1}{a}\ln \ln
\left( \frac{a}{x_{0}(1+\varepsilon (t))}e^{\frac{1}{2}a^{2}t}\right) +o(1)
\end{equation*}%
\begin{equation*}
=\frac{1}{a}\left( \ln \frac{x_{0}}{a}+\ln (1+\varepsilon (t))+\ln \left( 
\frac{1}{2}a^{2}t+\ln \frac{a}{x_{0}(1+\varepsilon (t))}\right) \right)
+o(1).
\end{equation*}%
Using that as $t\uparrow \infty $, $\ln (1+\varepsilon (t))=o(1)$ and that 
\begin{equation*}
\ln \left( \frac{1}{2}a^{2}t+\ln \frac{a}{x_{0}(1+\varepsilon (t))}\right)
=\ln \left( \frac{1}{2}a^{2}t\right) +o(1),
\end{equation*}%
assertion \eqref{t-expansion1} follows.

\subsubsection{Spatial asymptotic expansion of $h^{(-1)}(x,t_{0})$ for large 
$x$}

Let $t_{0}\geq 0.$ We show that, as $x\uparrow \infty ,$ 
\begin{equation}
h^{(-1)}(x,t_{0})=\frac{1}{2(1-\gamma )}t_{0}+\left( 1-\gamma \right) \left(
\ln x+\ln \ln x-\ln \frac{1}{1-\gamma }\right) +o(1).  \label{x-expansion1}
\end{equation}%
We first establish that 
\begin{equation}
\lim_{x\uparrow \infty }\frac{h^{(-1)}(x,t_{0})}{\ln x}=\left( 1-\gamma
\right) .  \label{xlim of h_inv}
\end{equation}%
Indeed, let $f(z,t):=\frac{1}{z}e^{\frac{1}{1-\gamma }z-\frac{1}{2}\left( 
\frac{1}{1-\gamma }\right) ^{2}t}$. Then,%
\begin{equation*}
\lim_{z\uparrow \infty }\frac{h(z,t_{0})}{f(z,t_{0})}=\lim_{z\uparrow \infty
}\int_{a}^{\frac{1}{1-\gamma }}ze^{z(y-\frac{1}{1-\gamma })-\frac{1}{2}%
(y^{2}-\left( \frac{1}{1-\gamma }\right) ^{2})t_{0}}dy
\end{equation*}%
\begin{equation*}
=\lim_{z\uparrow \infty }\left( \int_{a}^{\frac{1}{1-\gamma }%
}(z-yt_{0})e^{z(y-\frac{1}{1-\gamma })-\frac{1}{2}(y^{2}-\left( \frac{1}{%
1-\gamma }\right) ^{2})t_{0}}dy+\int_{a}^{\frac{1}{1-\gamma }}yt_{0}e^{z(y-%
\frac{1}{1-\gamma })-\frac{1}{2}(y^{2}-\left( \frac{1}{1-\gamma }\right)
^{2})t_{0}}dy\right)
\end{equation*}%
\begin{equation*}
=\lim_{z\uparrow \infty }\left( 1-e^{(a-\frac{1}{1-\gamma })z-\frac{1}{2}%
(a^{2}-\left( \frac{1}{1-\gamma }\right) ^{2})t_{0}}+\int_{a}^{\frac{1}{%
1-\gamma }}yt_{0}e^{z(y-\frac{1}{1-\gamma })-\frac{1}{2}(y^{2}-\left( \frac{1%
}{1-\gamma }\right) ^{2})t_{0}}dy\right) =1,
\end{equation*}%
where we used that $a<\frac{1}{1-\gamma }$ and, for the third term, the
monotone convergence theorem. Therefore, for each $t_{0}\geq 0,$ 
\begin{equation}
\lim_{x\uparrow \infty }\frac{h(x,t_{0})}{f\left( x,t_{0}\right) }=1.
\label{a-1}
\end{equation}

We now use a result on the inverses of asymptotic functions (see \cite%
{Entringer}) to prove the limit in (\ref{xlim of h_inv}) by verifying the
necessary assumptions for this result to hold. To this end, consider the
function $g(z):=\left( 1-\gamma \right) \ln z$, and notice that 
\begin{equation*}
g(f(z,t_{0}))=-\left( 1-\gamma \right) \ln z+z-\frac{1}{2\left( 1-\gamma
\right) }t_{0}\sim z,~~\  \text{as }~z\uparrow \infty .
\end{equation*}%
Thus, $\lim_{z\uparrow \infty }z^{-1}g(f(z,t_{0}))=1.$ Since, on the other
hand, $\lim_{z\uparrow \infty }f(z,t_{0})=\infty $, we deduce that $%
f^{(-1)}(x,t_{0})\sim g(x)$, as $x\uparrow \infty $. Moreover, $g(x)$ is
strictly increasing and the ratio $\frac{g_{x}\left( x,t_{0}\right) }{%
g\left( x,t_{0}\right) }\sim \frac{1}{x\ln x}=O(\frac{1}{x}),$ for
sufficiently large $x$. It then follows from the aforementioned result that $%
g(x)\sim h^{(-1)}(x,t_{0}),$ as $x\uparrow \infty ,$ and (\ref{xlim of h_inv}%
) follows.\newline

Next, we claim that, for each $t_{0}\geq 0$, 
\begin{equation}
\lim_{x\uparrow \infty }\frac{e^{\frac{1}{1-\gamma }(h^{(-1)}(x,t_{0})-\frac{%
1}{2}\frac{1}{1-\gamma }t_{0})}}{x\ln x}=1-\gamma .  \label{xlim}
\end{equation}%
Indeed, for $t_{0}=0$, we have from (\ref{h_inv 2}) that 
\begin{equation}
x=\int_{a}^{\frac{1}{1-\gamma }}e^{yh^{(-1)}(x,0)}dy=\frac{1}{h^{(-1)}(x,0)}%
\left( e^{\frac{1}{1-\gamma }h^{(-1)}(x,0)}-e^{ah^{(-1)}(x,0)}\right) ,
\label{A(x,0)}
\end{equation}%
and (\ref{xlim of h_inv}) yields that 
\begin{equation*}
\lim_{x\uparrow \infty }\frac{e^{\frac{1}{1-\gamma }h^{(-1)}(x,0)}}{x\ln x}
\end{equation*}%
\begin{equation*}
=\lim_{x\uparrow \infty }\frac{e^{\frac{1}{1-\gamma }h^{(-1)}(x,0)}}{e^{%
\frac{1}{1-\gamma }h^{(-1)}(x,0)}-e^{ah^{(-1)}(x,0)}}\frac{h^{(-1)}(x,0)}{%
\ln x}=1-\gamma .
\end{equation*}%
For $t_{0}>0$, we deduce from (\ref{h_inv 2})\ that 
\begin{equation}
x=\frac{1}{\sqrt{t_{0}}}e^{-\frac{h^{(-1)}(x,t_{0})^{2}}{2t_{0}}}\left( \Phi
\left( \frac{1}{1-\gamma }\sqrt{t_{0}}-\frac{h^{(-1)}(x,t_{0})}{\sqrt{t_{0}}}%
\right) -\Phi \left( a\sqrt{t_{0}}-\frac{h^{(-1)}(x,t_{0})}{\sqrt{t_{0}}}%
\right) \right) .  \label{x-long}
\end{equation}%
Then, using (\ref{aux1}), we have, for large $x,$

\begin{equation*}
1\leq \frac{1}{x\sqrt{t_{0}}}\exp \left( \frac{h^{(-1)}(x,t_{0})^{2}}{2t_{0}}%
\right) \Phi \left( \frac{1}{1-\gamma }\sqrt{t_{0}}-\frac{h^{(-1)}(x,t_{0})}{%
\sqrt{t_{0}}}\right)
\end{equation*}%
\begin{equation*}
\leq \frac{1}{x\sqrt{t_{0}}}e^{\frac{h^{(-1)}(x,t_{0})^{2}}{2t_{0}}}\frac{1}{%
\frac{h^{(-1)}(x,t_{0})}{\sqrt{t_{0}}}-\frac{1}{1-\gamma }\sqrt{t_{0}}}%
\varphi \left( \frac{1}{1-\gamma }\sqrt{t_{0}}-\frac{h^{(-1)}(x,t_{0})}{%
\sqrt{t_{0}}}\right)
\end{equation*}%
\begin{equation*}
=\frac{e^{\frac{1}{1-\gamma }(h^{(-1)}(x,t_{0})-\frac{1}{2}\frac{1}{1-\gamma 
}t_{0})}}{x(h^{(-1)}(x,t_{0})-\frac{1}{1-\gamma }t_{0})},
\end{equation*}%
and, in turn,%
\begin{equation*}
1\leq \liminf_{x\uparrow \infty }\left( \frac{e^{\frac{1}{1-\gamma }%
(h^{(-1)}(x,t_{0})-\frac{1}{2}\frac{1}{1-\gamma }t_{0})}}{xh^{(-1)}(x,t_{0})}%
\frac{h^{(-1)}(x,t_{0})}{h^{(-1)}(x,t_{0})-\frac{1}{1-\gamma }t_{0}}\right)
\end{equation*}%
\begin{equation*}
=\liminf_{x\uparrow \infty }\frac{e^{\frac{1}{1-\gamma }(h^{(-1)}(x,t_{0})-%
\frac{1}{2}\frac{1}{1-\gamma }t_{0})}}{xh^{(-1)}(x,t_{0})}\lim_{x\uparrow
\infty }\frac{h^{(-1)}(x,t_{0})}{h^{(-1)}(x,t_{0})-\frac{1}{1-\gamma }t_{0}}
\end{equation*}%
\begin{equation}
=\liminf_{x\uparrow \infty }\frac{e^{\frac{1}{1-\gamma }(h^{(-1)}(x,t_{0})-%
\frac{1}{2}\frac{1}{1-\gamma }t_{0})}}{xh^{(-1)}(x,t_{0})}.  \label{xliminf}
\end{equation}%
Similarly, we use that, $\ $for $~a<b<0,$ 
\begin{equation}
\Phi (b)-\Phi (a)\geq \frac{\varphi (a)-\varphi (b)}{a},\   \label{inequ3}
\end{equation}%
and deduce from (\ref{x-long}) that, for large $x,$%
\begin{equation*}
1\geq \frac{1}{x\sqrt{t_{0}}}e^{\frac{h^{(-1)}(x,t_{0})^{2}}{2t_{0}}}\frac{1%
}{a\sqrt{t_{0}}-\frac{h^{(-1)}(x,t_{0})}{\sqrt{t_{0}}}}
\end{equation*}%
\begin{equation*}
\times \left( \varphi (a\sqrt{t_{0}}-\frac{h^{(-1)}(x,t_{0})}{\sqrt{t_{0}}}%
)-\varphi (\frac{1}{1-\gamma }\sqrt{t_{0}}-\frac{h^{(-1)}(x,t_{0})}{\sqrt{%
t_{0}}})\right)
\end{equation*}%
\begin{equation*}
=\frac{e^{\frac{1}{1-\gamma }(h^{(-1)}(x,t_{0})-\frac{1}{2}\frac{1}{1-\gamma 
}t_{0})}}{{x(h^{(-1)}(x,t}_{0}{)-at_{0})}}-\frac{e^{a(h^{(-1)}(x,t_{0})-%
\frac{1}{2}at_{0})}}{{x(h^{(-1)}(x,t}_{0}{)-at_{0})}}.
\end{equation*}%
For the second term, we have 
\begin{equation*}
\lim_{x\uparrow \infty }\frac{e^{ah^{(-1)}(x,t_{0})}}{{x(h^{(-1)}(x,t}_{0}{%
)-at_{0})}}e^{-\frac{1}{2}at_{0}}=\lim_{x\uparrow \infty }\frac{%
e^{ah^{(-1)}(x,t_{0})-\ln x}}{{h^{(-1)}(x,t}_{0}{)-at_{0}}}e^{-\frac{1}{2}%
at_{0}}
\end{equation*}%
\begin{equation*}
=\lim_{x\uparrow \infty }\exp \left( a\ln x\left( \frac{h^{\left( -1\right)
}\left( x,t_{0}\right) }{\ln x}-\frac{1}{a}\right) \right) \frac{1}{{%
h^{(-1)}(x,t}_{0}{)-at_{0}}}=0.
\end{equation*}%
Therefore, 
\begin{equation*}
\lim \sup_{x\uparrow \infty }\left( \frac{e^{\frac{1}{1-\gamma }%
(h^{(-1)}(x,t_{0})-\frac{1}{2}\frac{1}{1-\gamma }t_{0})}}{{x(h^{(-1)}(x,t}%
_{0}{)-at_{0})}}-\frac{e^{a(h^{(-1)}(x,t_{0})-\frac{1}{2}at_{0})}}{{%
x(h^{(-1)}(x,t}_{0}{)-at_{0})}}\right)
\end{equation*}%
\begin{equation*}
=\lim \sup_{x\uparrow \infty }\frac{e^{\frac{1}{1-\gamma }(h^{(-1)}(x,t_{0})-%
\frac{1}{2}\frac{1}{1-\gamma }t_{0})}}{{x(h^{(-1)}(x,t}_{0}{)-at_{0})}}%
-\lim_{x\uparrow \infty }\frac{e^{a(h^{(-1)}(x,t_{0})-\frac{1}{2}at_{0})}}{{%
x(h^{(-1)}(x,t}_{0}{)-at_{0})}}
\end{equation*}%
\begin{equation*}
=\lim \sup_{x\uparrow \infty }\frac{e^{\frac{1}{1-\gamma }(h^{(-1)}(x,t_{0})-%
\frac{1}{2}\frac{1}{1-\gamma }t_{0})}}{{x(h^{(-1)}(x,t}_{0}{)-at_{0})}}
\end{equation*}%
\begin{equation*}
=\lim \sup_{x\uparrow \infty }\frac{e^{\frac{1}{1-\gamma }(h^{(-1)}(x,t_{0})-%
\frac{1}{2}\frac{1}{1-\gamma }t_{0})}}{x{h^{(-1)}(x,t}_{0}{)}}%
\lim_{x\uparrow \infty }\frac{x{h^{(-1)}(x,t}_{0}{)}}{x{(h^{(-1)}(x,t}_{0}{%
)-at_{0})}}
\end{equation*}%
\begin{equation}
=\limsup_{x\uparrow \infty }\frac{e^{\frac{1}{1-\gamma }(h^{(-1)}(x,t_{0})-%
\frac{1}{2}\frac{1}{1-\gamma }t_{0})}}{xh^{(-1)}(x,t_{0})}\leq 1.
\label{xlimsup}
\end{equation}%
From (\ref{xliminf}) and (\ref{xlimsup}), we then obtain%
\begin{equation*}
\limsup_{x\uparrow \infty }\frac{e^{\frac{1}{1-\gamma }(h^{(-1)}(x,t_{0})-%
\frac{1}{2}\frac{1}{1-\gamma }t_{0})}}{xh^{(-1)}(x,t_{0})}=1,
\end{equation*}%
which together with (\ref{xliminf}) gives (\ref{xlim}). Taking the logarithm
of both sides then yields 
\begin{equation*}
\lim_{x\uparrow \infty }\left( \frac{1}{1-\gamma }\left( h^{(-1)}(x,t_{0})-%
\frac{1}{2\left( 1-\gamma \right) }t_{0}\right) -\ln x-\ln \ln x\right) =\ln
\left( 1-\gamma \right) ,
\end{equation*}%
and the spatial asymptotic expansion (\ref{x-expansion1}) follows.

\subsubsection{Spatial asymptotics of $r(x,t_{0})$ for large $x$}

Let $t_{0}>0.$ We show that as $x\uparrow \infty ,$ the spatial asymptotic
expansion of $r(x,t_{0})$ is given by 
\begin{equation*}
r(x,t_{0})=\frac{1-\gamma }{t_{0}}x\ln \ln x+\frac{1}{t_{0}}\left( \left(
1-\gamma \right) x\ln x\right) ^{a\left( 1-\gamma \right) }e^{\frac{1}{2}a(%
\frac{1}{1-\gamma }-a)t_{0}}
\end{equation*}%
\begin{equation}
+\frac{1}{2\left( 1-\gamma \right) }x-\frac{1-\gamma }{t_{0}}x\ln \frac{1}{%
1-\gamma }+o(1).  \label{space expansion of r}
\end{equation}%
Indeed, from (\ref{r-h}) and (\ref{h1}), we have 
\begin{equation*}
r(x,t_{0})=\int_{a}^{\frac{1}{1-\gamma }}\frac{1}{t_{0}}%
(yt_{0}-h^{(-1)}(x,t_{0}))e^{yh^{(-1)}(x,t_{0})-\frac{1}{2}y^{2}t_{0}}dy
\end{equation*}%
\begin{equation*}
+\frac{h^{(-1)}(x,t_{0})}{t_{0}}\int_{a}^{\frac{1}{1-\gamma }%
}e^{yh^{(-1)}(x,t_{0})-\frac{1}{2}y^{2}t_{0}}dy
\end{equation*}%
\begin{equation*}
=\frac{1}{t_{0}}\left( e^{a(h^{(-1)}(x,t_{0})-\frac{1}{2}at_{0})}-e^{\frac{1%
}{1-\gamma }(h^{(-1)}(x,t_{0})-\frac{1}{2(1-\gamma )}t_{0})}\right) +\frac{%
h^{(-1)}(x,t_{0})}{t_{0}}x,
\end{equation*}%
where we used (\ref{x-h})\ for the last term. Then, (\ref{space expansion of
r}) follows using (\ref{x-expansion1}).\newline

For $t_{0}=0$, we have from (\ref{A(x,0)}) that 
\begin{equation*}
r(x,0)=\frac{1}{1-\gamma }x-\frac{(\frac{1}{1-\gamma }-a)e^{ah^{(-1)}(x,0)}-x%
}{h^{(-1)}(x,0)},
\end{equation*}%
and, for large $x,$ 
\begin{equation}
r(x,0)=\frac{1}{1-\gamma }x\left( 1-\frac{1}{\ln x}\right) +o(1).
\label{space expansion of r for t=0}
\end{equation}%
From (\ref{space expansion of r}) and (\ref{space expansion of r for t=0}),
we then obtain that for $t_{0}>0$ and $t_{0}=0,$ we have respectively, 
\begin{equation*}
r(x,t_{0})\sim \frac{1-\gamma }{t_{0}}x\ln \ln x\text{ \  \  \  \ and \  \  \ }%
r(x,0)\text{ }\sim \text{ }\frac{1}{1-\gamma }x.
\end{equation*}%
Therefore, the risk tolerance function does \textit{not }have the spatial
turnpike property (\ref{x-lim of r(x,t)}). Recall that the underlying
measure \textit{lacks} a Dirac mass on the right boundary of the measure $%
\mu ,$ which is a necessary condition for the results in Proposition 3 to
hold.

\begin{itemize}
\item The case $a=0^{+}$
\end{itemize}

We conclude with the case that $\mu $ is the Lebesgue measure on $(0,\frac{1%
}{1-\gamma }]$. For $t_{0}\geq 0,$ we easily obtain the same spatial
asymptotic expansions of $h^{(-1)}(x,t_{0})$ as in (\ref{x-expansion1}) and
of $r(x,t_{0})$ as in (\ref{space expansion of r}) and (\ref{space expansion
of r for t=0}).

For the temporal expansion, we claim that as $t\uparrow \infty ,$ 
\begin{equation}
\frac{h^{(-1)}(x_{0},t)}{t}=\frac{\sqrt{\ln t+2\ln x_{0}-\ln 2\pi }}{\sqrt{t}%
}+o(\frac{1}{\sqrt{t}}).  \label{final}
\end{equation}%
To see this, first recall (cf. (\ref{x-h})) that 
\begin{equation}
x_{0}=\int_{0^{+}}^{\frac{1}{1-\gamma }}e^{y(h^{(-1)}(x_{0},t)-\frac{1}{2}%
yt)}dy,  \label{aux-3}
\end{equation}%
Taking the logarithm of both sides of (\ref{aux-3}) yields 
\begin{equation}
2\ln x_{0}=\left( \frac{h^{(-1)}(x_{0},t)}{\sqrt{t}}\right) ^{2}-\ln t
\label{aux-4}
\end{equation}%
\begin{equation*}
+2\ln \left( \Phi \left( \sqrt{t}(\frac{1}{1-\gamma }-\frac{h^{(-1)}(x_{0},t)%
}{t})\right) -\Phi \left( -\frac{h^{(-1)}(x_{0},t)}{\sqrt{t}}\right) \right)
.
\end{equation*}%
Next, we claim that $l:=\liminf_{t\uparrow \infty }\frac{h^{(-1)}(x_{0},t)}{%
\sqrt{t}}=\infty .$ Indeed, if $l<\infty $, then, as $t\uparrow \infty $,
the above yields 
\begin{equation*}
2\ln x_{0}=l^{2}-\lim_{t\uparrow \infty }\left( \ln t\right) +2\ln (1-\Phi
(-l))=-\infty ,
\end{equation*}%
which is a contradiction. Therefore, it must be that $l=\infty ,$ which
combined with the fact that $\lim_{t\uparrow \infty }\frac{h^{\left(
-1\right) }\left( x_{0},t\right) }{t}=0,$ implies that as $t\uparrow \infty $%
, the third term on the right hand side of (\ref{aux-4}) converges to $2\ln 
\sqrt{2\pi }$. Thus, we obtain 
\begin{equation*}
2\ln x_{0}=\lim_{t\uparrow \infty }\left( \left( \frac{h^{(-1)}(x_{0},t)}{%
\sqrt{t}}\right) ^{2}-\ln t+2\ln \sqrt{2\pi }\right) ,
\end{equation*}%
from which we deduce that $h^{(-1)}(x_{0},t)=\sqrt{t(\ln t+2\ln x_{0}-\ln
2\pi )}+o(\sqrt{t}),$ and (\ref{final}) follows.

\section{Extensions}

We have analyzed the spatial and temporal asymptotic behavior of the risk
tolerance function $r\left( x,t\right) $. We recall that the optimal
portfolio process $\pi _{t}^{\ast ,x}$ is given in the feedback form $\pi
_{t}^{\ast ,x}=\sigma _{t}^{+}\lambda _{t}r\left( X_{t}^{\ast ,x},t\right) ,$
with $X_{t}^{\ast ,x}$ being the wealth generated by it. Furthermore, it was
shown in \cite{MZ10a} that $X_{t}^{\ast ,x}$ and $\pi _{t}^{\ast ,x}$ are
given in the closed form 
\begin{equation*}
X_{t}^{\ast ,x}=h\left( h^{\left( -1\right) }\left( x,0\right)
+A_{t}+M_{t},A_{t}\right) ,\text{ }\pi _{t}^{\ast ,x}=\sigma _{t}^{+}\lambda
_{t}h_{x}\left( h^{\left( -1\right) }\left( x,0\right)
+A_{t}+M_{t},A_{t}\right) .
\end{equation*}%
It is then natural to investigate the long-term limits $\lim_{t\uparrow
\infty }X_{t}^{\ast ,x},$ $\lim_{t\uparrow \infty }\pi _{t}^{\ast ,x}$ under
asymptotic assumptions on the initial datum and the results obtained herein.
The asymptotic behavior of these processes has been investigated in \cite%
{GKRX14} for the classical setting.

In a different direction, an interesting problem is how to construct
investment policies which yield a targeted long-term wealth distribution. In
a static model, this question was analyzed in \cite{Sharpe} and in the
log-normal, classical and forward cases, in \cite{Monin}. However, in these
settings, there is a strong model commitment, which is a nonrealistic
assumption for long-term portfolio management.

In the forward setting we have analyzed herein, the model is dynamically
updated. Furthermore, the distribution of the optimal wealth is given
explicitly, using the above formula, by 
\begin{equation*}
\mathbb{P}\left( X_{t}^{\ast ,x}\leq y\right) =\mathbb{P}\left( h^{\left(
-1\right) }\left( x,0\right) +A_{t}+M_{t}\leq h^{\left( -1\right) }\left(
y,A_{t}\right) \right) 
\end{equation*}%
\begin{equation*}
=\mathbb{P}\left( \frac{h^{\left( -1\right) }\left( x,0\right) }{%
\left \langle M\right \rangle _{t}}+1+\frac{M_{t}}{\left \langle M\right \rangle
_{t}}\leq \frac{h^{\left( -1\right) }\left( y,A_{t}\right) }{\left \langle
M\right \rangle _{t}}\right) ,
\end{equation*}%
where we used that $A_{t}=\left \langle M\right \rangle _{t}$ (cf. (\ref%
{market input})). Therefore, one expects that the limit (\ref{limit-a/2}) as
well as results on strong law of large numbers for martingales can be used
to study the long-term distribution of the optimal processes. Such questions
are currently investigated by the authors in \cite{geng-z} and others.





\end{document}